\documentclass{article}

\usepackage{PRIMEarxiv}

\usepackage[utf8]{inputenc} % allow utf-8 input
\usepackage[T1]{fontenc}    % use 8-bit T1 fonts
\usepackage{hyperref}       % hyperlinks
\usepackage{url}            % simple URL typesetting
\usepackage{booktabs}       % professional-quality tables
\usepackage{amsfonts}       % blackboard math symbols
\usepackage{nicefrac}       % compact symbols for 1/2, etc.
\usepackage{microtype}      % microtypography
\usepackage{lipsum}
\usepackage{fancyhdr}       % header
\usepackage{graphicx}       % graphics
\usepackage{subfigure}
\graphicspath{{media/}}     % organize your images and other figures under media/ folder

\usepackage{amssymb}
\usepackage{amsthm}
\usepackage{amssymb}
\usepackage{amsfonts} % 数学字体
\usepackage{amsmath} 
\usepackage[ruled,linesnumbered]{algorithm2e}
\usepackage{rotating}

\newtheorem{theorem}{Theorem}
\newtheorem{lemma}{Lemma}

\newtheorem{example}{Example}

\title{Adaptive Approximate Implicitization of Planar Parametric Curves via Weak Gradient Constraints
}

\author{
  Minghao Guo \\
  School of Mathematics and Statistics \\
  Changchun University of Technology \\
  Changchun\\
  \texttt{mhguo@ccut.edu.cn} \\
   \And
  Yan Gao \\
  School of Artificial Intelligence \\
  Jilin University \\
  Changchun\\
  \AND
  Zheng Pan \\
  Chang Guang Satellite Technology Company Ltd. \\
  Changchun \\
}

\begin{document}
\maketitle

\begin{abstract}
Converting a parametric curve into the implicit form, which is called implicitization, 
has always been a popular but challenging problem
in geometric modeling and related applications.
However, the existing methods mostly suffer from the problems
of maintaining geometric features and choosing a reasonable 
implicit degree. The present paper has two
contributions. We first introduce a new regularization constraint
(called the weak gradient constraint) for 
both polynomial and non-polynomial curves, which efficiently possesses shape preserving. 
We then propose two adaptive algorithms of approximate implicitization for 
polynomial and non-polynomial curves respectively, which find the ``optimal'' implicit degree 
based on the behavior of the weak gradient constraint.
More precisely, the idea is 
gradually increasing the implicit degree, until there is no obvious improvement 
in the weak gradient loss of the outputs.
Experimental results have shown the effectiveness and high quality of our proposed methods.
\end{abstract}

% keywords can be removed

\section{Introduction}
In geometric modeling and computer aided design, 
the implicit and the parametric form are two main 
representations of curves and surfaces.
The parametric representations provide a simple way of 
generating points for displaying curves and surfaces.
However, it has to introduce a parametrization of the geometry, 
which is always a challenging problem.
Without requiring any parametrization, 
the implicit representations offer a number of advantages,
such as the closeness under certain geometric operations (like union, intersection, and blending),
the representation ability for describing an object with complicated geometry, and so on.
Thus in this paper, we discuss the problem of
converting the parametric form of curves
into the implicit form, which is called implicitization.

Implicitization has been receiving
increased attention in the past few years. 
Traditional implicitization approaches are based on the elimination theory 
(such as $\mu$-basis\cite{perez2021inversion}, Gr\"{o}bner bases\cite{tran2004efficient}\cite{anwar2021determining}, resultants\cite{perez2008univariate} and moving curves and surfaces\cite{buse2019implicitizing}\cite{lai2019implicitizing}),
in which the implicitization problem is solved by elimination 
of the parametric variables. However, high polynomial degrees of
their outputs not only make this form computationally
expensive and numerical unstable, but also cause
self-intersections and unwanted branches in most cases.

To alleviate this problem, a number of approximate
implicitization techniques have been proposed, 
and we follow this line of work.
Most of these methods fix the degree of the objective
implicit form, and thus the implicitization problem
converts to find a solution in a finite vector space. 
We group these methods into two categories:

Methods in the first category minimize the algebraic distance
from the input parametric curve/surface to the output implicit curve/surface, 
with a chosen implicit degree. One of the first methods 
that use this idea is \cite{dokken2001approximate}, in which
the main approximation tool is the singular value decomposition.
Later, \cite{barrowclough2012approximate} discussed theoretical 
and practical aspects of the Dokken's method in \cite{dokken2001approximate} 
under different polynomial basis functions,
and proposed a new method for a least squares approach to approximate implicitization 
using orthogonal polynomials, see Section \ref{dokkenmethod} for details.
Furthermore, piecewise approximate implicitization with prescribed interpolating conditions 
using tensor-product B-splines was studied by \cite{Raffo2019}.
Recently, \cite{raffo2020reverse} proposed a method to determine primitive shapes of geometric models
by combining clustering analysis with the Dokken's method. In that paper,
the implicit degree of curve/surface patches was determined by checking whether the smallest 
singular value is less than a certain threshold, where the threshold was
inferred using a straightforward statistical approach.

The second category is the fitting-based methods,
which is also called discrete approximate implicitization.
These methods minimize the squared algebraic distances of a set of points 
sampled from the given parametric curves and surfaces.
\cite{wang2007approximate} proposed an approach, that
is to approximate the set
of sampling points with MQ quasi-interpolation in order to possess shape preserving and then
to approximate the error function by using RBF networks.
\cite{wang2021encoder} developed an autoencoder-based fitting method, 
which is to put sampling points into an encoder to obtain polynomial coefficients 
and then put them into a decoder to output the predicted function value.
\cite{juttler2000least} described an algorithm for approximating sampling points 
and associated normal vectors simultaneously, in the context of fitting 
with implicit defined algebraic spline curves. In that paper,
the coefficients of the output function $f$ are obtained as the minimum of
\begin{equation*}
    \sum_{j=1}^{N} (f(\mathbf{p}_j))^2+\lambda
    \sum_{j=1}^{N} 
    \left \| \bigtriangledown f(\mathbf{p}_j)- \mathbf{n}_j \right \|^2
    +\mbox{``tension''},
\end{equation*}
where $\left \{ \mathbf{p}_j \right \} _{j=1}^N$ are sampled point data, 
$\left \{ \mathbf{n}_j \right \} _{j=1}^N$ are unit normals at these points,
and $\lambda$ is the regulator gain.
The first term represents the (algebraic) distance of the point data from the implicit curve
$f=0$. The second term, which is called the \textbf{strong gradient constraint} by us in this paper,
controls the influence of the normal vectors $\mathbf{n}_j$ to the resulting curve.
However, the strong gradient constraint is too strict to find output curves of low degrees,
since it requires that the gradient vector
$\bigtriangledown f(\mathbf{p}_j)$ and the normal vector $\mathbf{n}_j$
have same direction and magnitude simultaneously in any point $\mathbf{p}_j$.
The third term $\mbox{``tension''}$ is added in order to pull the approximating curve towards a simpler shape.
Afterwards, the idea of \cite{juttler2000least} has been generalized to the algebraic spline surface case in 
\cite{B2002Least}, the rotational surface case in \cite{shalaby2008approximate},
and the space curve case in \cite{aigner2012approximate}.

A related approach has been proposed by \cite{interian2017curve}, where an alternative gradient constraint 
\begin{equation*}
    1-\frac{1}{N}\sum_{j=1}^{N} 
    \frac{\bigtriangledown f(\mathbf{p}_j)}{\left \| \bigtriangledown f(\mathbf{p}_j) \right \| } 
    \cdot \mathbf{n}_j
\end{equation*}
is defined to make the normalized gradient vector
$\tfrac{\bigtriangledown f(\mathbf{p}_j)}{\left \| \bigtriangledown f(\mathbf{p}_j) \right \| }$
close to the unit normal vector $\mathbf{n}_j$ in any point $\mathbf{p}_j$.
More precisely, this gradient constraint requires that the mean of 
angles of $\tfrac{\bigtriangledown f(\mathbf{p}_j)}{\left \| \bigtriangledown f(\mathbf{p}_j) \right \| }$
and $\mathbf{n}_j$ equals to zero.
This constraint is used subsequently in their algorithm 
to find the ``optimal'' degree of the implicit polynomial needed for the representation of the data set.
However, the normalization of both gradient and normal vectors leads to multimodal functional 
dependencies of the solution from data, and computationally expensive metaheuristic algorithms have to be
used in order to realize a better exploration of the search space.

In this paper, we attempt to recover the approximate implicitization
of the planar parametric curve adaptively.
The contributions of our work are summarized as follows:
% for both polynomial and non-polynomial curves,
\begin{itemize}
    \item To tackle challenges such as unwanted branches and computational complexity,
    we introduce the so-called \textbf{weak gradient constraint}, 
    which bends the direction of the implicit curve closer to that of the parametric curve.
    Moreover, compared to the strong gradient constraint in \cite{juttler2000least}, 
    our new regularization constraint largens the solution space.
    \item We perform our objective function of approximate implicitization into the
    quadratic form, so that the eigenvalue/eigenvector method can be used to find the minimum rapidly.
    \item We develop an adaptive implicitization algorithm, 
    which is to find an implicit polynomial that produces a compact and smooth representation 
    of the input curve with the lowest degree as possible, 
    and at the same time minimizes the implicitization error.
\end{itemize}
% • We define two single-layer RNNs, Curve RNN and Point RNN that can be used to perform curve predictions and control point predictions respectively. • We provide a Hierarchical RNN architecture model that performs reconstruction of 2D spline curves of variable number, each of which contains a variable number of control points using nested Curve RNN and Point RNN units, and an algorithm to train it effectively. • We provide an unsupervised parametric reconstruction model that performs reconstruction of 3D surfaces of extrusion or revolution.
% To summarize, our contributions are as follows:
% We propose a single-step method specializing in parameterizable coherent signal tracing as a substitute for the existing multi-step procedure.

% We discuss the continuity, differentiability, and convexity of our objective function systemically via visualization.

% We demonstrate that our method outperforms competing approaches under various SNR or lighting conditions across several application domains.

% The validation of the method on both synthetic and real data.

The remainder of the paper is organized as follows.
We state the problem and present a synopsis of the Dokken's method 
for approximate implicitization in Section \ref{background}. Section \ref{polynomial method} introduces our implicitization 
method (WGM) for polynomial curves and shows our numerical results. In Section \ref{non polynomial method}, the WGM is extended to non-polynomial curves. Section \ref{conclusion} finalizes the paper with a conclusion and some possible directions for future work.

\section{Background}\label{background}
\subsection{Problem formulation}
A parametric polynomial curve of degree $m$ in $\mathbb{R}^2$ is given by
\begin{equation*}
    \mathbf{p}(t)=\left ( \begin{matrix}
 p_1(t)\\
p_2(t)
\end{matrix} \right ), t\in [a,b],  
\end{equation*}
where $p_1$ and $p_2$ are polynomials in $t$.
An implicit (algebraic) curve of degree $n$ in $\mathbb{R}^2$, is defined by the
zero contour of a bivariate polynomial
\begin{align*}
    f_{\mathbf{b}}(x,y) & =
    \sum_{i=1}^kb_{i}\phi_i(x,y)\\
%     \sum_{0\le i+j\le n}b_{ij}x^iy^j  \\
%                         & = \left ( \begin{matrix}
%  1 & x & y \cdots & y^n
% \end{matrix} \right ) \left ( \begin{matrix}b_{00}
%  \\b_{10}
%  \\b_{01}
%  \\\vdots 
%  \\b_{0n}
% \end{matrix} \right ) \\
& =\left ( \phi_1(x,y), \phi_2(x,y), \ldots ,
\phi_k(x,y)\right )
\left ( \begin{matrix}b_1
 \\b_2
 \\\vdots
 \\b_k
\end{matrix} \right ),
\end{align*}
where $\left \{ \phi_i(x,y) \right \}_{i=1}^{k}$
generates a basis for bivariate polynomials of total degree
$n$, $k=\left ( \begin{matrix}n+2
 \\2
\end{matrix} \right )$ denotes the number of basis
functions,
$\mathbf{b}=\left ( b_1,b_2,\ldots ,b_k \right )^\top $ 
is the coefficient vector of $f_{\mathbf{b} }(x,y)$.

An exact implicitization of 
$\mathbf{p}(t)$ is a non-zero 
$f_{\mathbf{b}}(x,y)$,
such that
% the pointwise error
% \begin{equation}
%     \left \| f_{\mathbf{b}}\circ \mathbf{p} \right \|_{\infty}
% =\max_{t\in[a,b]}\left | f_{\mathbf{b}}(\mathbf{p}(t)) \right | =0,
% \end{equation}
% or 
the squared algebraic distance (AD for short) from $\mathbf{p}(t)$ to 
the implicit curve $f_{\mathbf{b}}(x,y)=0$ equals to zero, i.e.
\begin{equation*}
    L_{AD}=
    \int_{a}^{b} \left [ f_{\mathbf{b}}(\mathbf{p}(t)) \right ]^2 \mathrm{d} t =0.
\end{equation*}
However, as stated in the introduction, 
in practice one may prefer 
to work with lower degrees.
% since the number of coefficients 
% may be very high otherwise.
Thus in this paper, we consider the approximate
implicitization problem, which is to seek the ``optimal''
$f_{\mathbf{b}}(x,y)$ with a lower degree $n$, 
that minimizes the squared AD constraint $L_{AD}$
% \begin{equation}
%     \min_{f_\mathbf{b}}\left \{
%     \left \| f_{\mathbf{b}}\circ \mathbf{p} \right \|_{2}
%     +\mbox{regularizing terms} \right \}.  
% \end{equation}
%where $\left \| \cdot  \right \|_{2} $ is $L_2$ norm.
% approach
% the given $\mathbf{p}(t)$ 
under some additional criterion to be specified.

\subsection{Dokken's (weak) method}\label{dokkenmethod}
% 其中approximate method分为两类. 
% 一类是连续情形:

% Sederberg: By using monoid curves and surfaces, the
% method eliminates the undesirable singularities and
% “phantom” branches normally associated with implicit
% representation. 

% Juttler的基于插值的方法等等

% 另一类是discrete情形, i.e. the curve/surface fitting.

% The squared algebraic distance of parametric curve
% $\mathbf{p}(t)$ from implicit curve $f_{\mathbf{b} }(x,y)=0$
% is given by 
% \begin{equation} 
% \label{al_dis}
%     \int_{a}^{b} \left [ f_{\mathbf{b}}(\mathbf{p}(t)) \right ]^2  dt,
% \end{equation}
In this subsection, we give a brief description of the Dokken's (weak) method.
Notice that the expression $f_{\mathbf{b}}(\mathbf{p}(t))$ is a univariate
polynomial of degree $mn$ in $t$, 
% Then the approximate
% implicitization problem is converted to find a
% $f_{\mathbf{b}}(\mathbf{p}(t))$
% that min (\ref{al_dis}). 
Dokken finds that $f_{\mathbf{b}}(\mathbf{p}(t))$  can be factorized as
\begin{equation*}
    f_{\mathbf{b}}(\mathbf{p}(t))=
    (\alpha(t))^\top D_1\mathbf{b},
\end{equation*}
where 
\begin{itemize}
    \item $\mathbf{b}$ is the unknown coefficient vector of $f_{\mathbf{b} }(x,y)$,
    \item $\alpha(t)=\left ( \alpha _1(t),\alpha _2(t),\ldots ,\alpha _{mn+1}(t) \right )^\top$
    is the basis of the space of univariate polynomials of degree $mn$, and
    \item $D_1$ is the collocation matrix whose columns are the coefficients 
of $\phi _i(\mathbf{p}(t))$ expressed in the $\alpha(t)$-basis.
\end{itemize}
\begin{lemma}\label{lemma_dokken}\cite{barrowclough2012approximate}
Let \begin{equation}\label{gram}
     G_{\alpha}=
     \int_{a}^{b} \alpha(t)(\alpha(t))^\top \mathrm{d} t
\end{equation}
denote the Gram matrix of the basis $\alpha(t)$.
Then the squared AD of $\mathbf{p}(t)$ 
from $f_{\mathbf{b} }(x,y)=0$ can be written as
\begin{equation}\label{l_ad}
L_{AD}=
\int_{a}^{b} \left [ f_{\mathbf{b}}(\mathbf{p}(t)) \right ]^2 \mathrm{d} t=\mathbf{b}^\top A_1 \mathbf{b},
\end{equation}
where 
\begin{equation}
A_1 = D_1^\top G_{\alpha}D_1
\end{equation}
is a positive semidefinite matrix.
\end{lemma}

Lemma \ref{lemma_dokken} shows that $L_{AD}$ is a homogeneous quadratic form of $\mathbf{b}$.
In order to avoid the null vector $\mathbf{b}=0$, Dokken introduce the normalization
$\left \| \mathbf{b} \right \| =1$. Denote by $\mathbf{b}_{DM}$ the unit eigenvector 
corresponding to the smallest eigenvalue of $A_1$,
then $\mathbf{b}_{DM}$ is the solution of the Dokken's method for minimizing 
$L_{AD}$
subject to $\left \| \mathbf{b} \right \| =1$.
% \begin{remark}
% %可以得到精确参数化
% %空间曲线的隐式化中的一段话
% Dokken’s weak method – when combined with numerical integration for evaluating the objective function

% hese three approaches are able to provide meaningful solutions which minimize the squared algebraic distances (10). However, they may still lead to fairly unexpected results. Additional branches and isolated singular points may be present, even for data which are sampled from regular curves or surfaces.
% \end{remark}
\section{Methodology for Polynomial Curves}\label{polynomial method}
In this section, we provide the approximate implicitization methodology for polynomial curves. 
First, we propose the weak gradient constraint to keep
the gradient vector of $f_{\mathbf{b}}(x,y)=0$ and 
the tangent vector of $\mathbf{p}(t)$ being perpendicular.
Then, we represent the objective function into the matrix form.
Finally, we propose the adaptive implicitization algorithm 
to compute the ``optimal'' implicitization $f_{\mathbf{b}}(x,y)$
and do some experiments to show the validity of the algorithm.
\subsection{Distance constraint}
We use the squared AD constraint in Equation (\ref{l_ad}):
% $f_{\mathbf{b}}(x,y)=0$ to
% $\mathbf{p}(t)$:
\begin{equation*} 
\label{al_dis}
    L_{AD}
=\mathbf{b}^\top A_1 \mathbf{b}.
\end{equation*}
\subsection{Weak gradient constraint}
To obtain a non-trivial solution, 
the implicitization problem must be regularized by
restricting $f_{\mathbf{b}}$ to some specified class 
of functions. One reasonable approach is to require 
that this be the class of ``shape-preserving'' 
functions. 
We present the so-called weak gradient (WG for short) constraint:
\begin{equation}\label{lwg}
    L_{WG}=\left \| \bigtriangledown f_{\mathbf{b}}\cdot{\mathbf{p}}'  \right \|
    =\int_{a}^{b}\left [ \bigtriangledown f_{\mathbf{b}}(\mathbf{p}(t))\cdot{\mathbf{p}' (t)} \right ]^2 \mathrm{d} t, 
\end{equation}
where 
\begin{itemize}
    \item $\bigtriangledown f_{\mathbf{b}}(\mathbf{p}(t))$ is the
gradient vector of the implicit curve at the point $\mathbf{p}(t)$,
    \item ${\mathbf{p}'(t)}$ is the tangent 
vector of the parametric curve at the point $\mathbf{p}(t)$, and 
    \item the inner product
    \begin{equation*}
    \bigtriangledown f_{\mathbf{b}}(\mathbf{p}(t))\cdot{\mathbf{p}' (t)}
    =\left \| \bigtriangledown f_{\mathbf{b}}(\mathbf{p}(t)) \right \|
    \left \| {\mathbf{p}' (t)} \right \| \cos\theta,
    % \left (
    % \widehat{\bigtriangledown f_{\mathbf{b}},\mathbf{p}'}  \right ). 
\end{equation*}
where $\theta$ denotes  the angle of $\bigtriangledown f_{\mathbf{b}}(\mathbf{p}(t))$ 
and $\mathbf{p}' (t)$ at the point $\mathbf{p}(t)$.
\end{itemize}
  
Compared to the strong gradient constraint in \cite{juttler2000least},
Our WG constraint only requires that the gradient vector 
of the implicit curve and 
the normal vector of the parametric curve
have same direction in any point.
Intuitively speaking, the WG constraint bend the direction of
the implicit curve closer to 
the parametric curve's.
If the inner product $\bigtriangledown f_{\mathbf{b}}(\mathbf{p}(t))\cdot{\mathbf{p}' (t)}$
equals $0$, then the tangents of the implicit and parametric curve
are exactly parallel.
% , certainly for regular curves.
The smaller the inner product is, the more similar are the
appearance of them.
% \begin{remark}
% strong gradient constraint??
% %扩大解空间
% %弱约束便于用特征值求解
% \end{remark}

\begin{theorem}
The WG constraint $L_{WG}$ in Equation (\ref{lwg}) can be written in 
a homogeneous quadratic form of $\mathbf{b}$ using the
basis $\alpha(t)$.
% \begin{equation}
%     L_{WG}=\mathbf{b}^TD_2^TD_2\mathbf{b}=\mathbf{b}^TA_2\mathbf{b}.
% \end{equation}
\end{theorem}
%% Example of a proof:
\begin{proof}
We can perform the WG constraint as follows.
Since the expressions ${\mathbf{p}' (t)}$ and
$\bigtriangledown f_{\mathbf{b}}(\mathbf{p}(t))$
are polynomial vectors of degree $m-1$ and $(n-1)m$ in
$t$ respectively, their inner product $\bigtriangledown f_{\mathbf{b}}(\mathbf{p}(t))\cdot{\mathbf{p}' (t)}$ 
is a polynomial of degree $nm-1$ in $t$. Thus it also can be written as a 
linear combination of $\alpha(t)$, which is the basis for univariate polynomials of degree $mn$.
Every coefficient of this linear combination
is a linear expression of $\mathbf{b}$.
As a result, the inner product can
be factored into 
\begin{align*}
    \bigtriangledown f_{\mathbf{b}}(\mathbf{p}(t))\cdot{\mathbf{p}' (t)}
    &=({\alpha}(t))^\top
    \left ( \begin{matrix}
\mathbf{b}\mbox{'s linear expression}\\
\vdots \\
\mathbf{b}\mbox{'s linear expression}
\end{matrix} \right ) \\
    &=({\alpha}(t))^\top D_2\mathbf{b},
\end{align*}
where $D_2$ 
is the collocation matrix whose rows are the coefficients of 
``$\mathbf{b}$'s linear expressions'' expressed
in $\mathbf{b}$.
Finally, the WG constraint can be written as
\begin{equation}\label{eq_lwg}
 L_{WG}
 =\int_{a}^{b}\left [ \bigtriangledown f_{\mathbf{b}}(\mathbf{p}(t))\cdot{\mathbf{p}' (t)} \right ]^2 \mathrm{d} t
 =\mathbf{b}^\top A_2 \mathbf{b},
\end{equation}
where 
\begin{equation}
A_2 = D_2^\top G_{\alpha}D_2
\end{equation}
is a positive semidefinite matrix and $G_{\alpha}$
is the Gram matrix of the basis $\alpha(t)$ in Equation (\ref{gram}).
\end{proof}
% penalizes the “roughness”, measured by the average squared curvature, of the model. The trade-off between these competing requirements is controlled by 
% called the regulator gain or smoothing parameter.

\subsection{Putting things together}
%$\mathbf{b}$ is the unknown 
Summing up, due to the Equation(\ref{l_ad}) and (\ref{eq_lwg}),
the approximate
implicitization is found by minimizing 
the positive semidefinite quadratic objective function
\begin{equation}\label{obfun}
    L_{\lambda,n}(\mathbf{b})=
    L_{AD} +\lambda L_{WG}
    =\mathbf{b}^\top (A_1+ \lambda A_2) \mathbf{b}
\end{equation}
over the coefficients $\mathbf{b}$ of
$f_{\mathbf{b} }(x,y)$,  while keeping the 
degree $n$ of $f_{\mathbf{b} }(x,y)$ fixed. 
The first term in Equation (\ref{obfun}) 
measures the fidelity of the implicit curve to the given parametric curve, 
and the second term in Equation (\ref{obfun}) try to 
maintain geometric features that the implicit curve must have. 
The trade-off between these requirements is controlled by $\lambda>0$, called the regulator gain.

Similar to the Dokken's Method, 
denote by $\mathbf{b}_{WGM}$ the unit eigenvector 
corresponding to the smallest eigenvalue of
\begin{equation*}
    A=A_1+ \lambda A_2,
\end{equation*}
then $\mathbf{b}_{WGM}$ is the solution for minimizing Equation (\ref{obfun})
% $L_{\lambda,n}(\mathbf{b})$
subject to $\left \| \mathbf{b} \right \| =1$.
% In this paper, we choose
% \begin{equation}
%     R(\mathbf{b})=
%   \int_{a}^{b} 
%   \left[f_{\mathbf{b}}(\mathbf{p}(t))\right]^2 dt 
% \end{equation}
% for reasons explained below.
% 这里需要解释为何不使用几何距离.
\subsection{Adaptive implicitization algorithm}
The adaptive implicitization is to obtain 
the ``optimal'' degree $n_{op}$ for the implicit polynomial
$f_{\mathbf{b}}$, where $1 \le n_{op} \le n_{\max}$.
We estimate $n_{op}$ via the behavior of the WG constraint
as the implicit degree $n$ increases.
We have done lots of experiments on examining the 
change trend of the WG's loss,
and find that the change usually goes through three stages
as $n$ increases:
\begin{itemize}
    \item First, the WG's loss drops significantly (i.e. underfitting).
    \item Second, the WG's loss reaches the minimum, and then changes very slightly
    (i.e. justfitting).
    \item Third, the WG's loss increases conversely (i.e. overfitting).
\end{itemize}
Thus, we introduce two thresholds for the stopping criterion:
\begin{itemize}
    \item $\epsilon_{AD}$: to examine whether the AD's loss 
    in Equation (\ref{l_ad}) satisfies our default precision;
    \item $\epsilon_{WG}$: to check the monotonicity of the WG's loss
    in Equation (\ref{eq_lwg}) to avoid overfitting.
\end{itemize}
\begin{algorithm}
\caption{WGM for polynomial curves}\label{alg:two}
\KwData{The polynomial curve $\mathbf{p}(t)$;
the maximum implicit degree $n_{\max}$;
the stopping thresholds $\epsilon_{AD},\epsilon_{WG}$.}
\KwResult{The coefficient vector $\mathbf{b}_{WGM}$ of the implicit 
polynomial $f_{\mathbf{b}}$.}
$n\leftarrow 1$\;
\While{$n\leq n_{\max}$}{
Construct the collocation matrix $D_1^{(n)}$ 
of the AD constraint\;
Construct the collocation matrix $D_2^{(n)}$ 
of the WG constraint\;
Compute the Gram matrix $G_{\alpha}^{(n)}$\;
$A_1^{(n)}, A_2^{(n)}\leftarrow 
(D_1^{(n)})^\top G_{\alpha}^{(n)} D_1^{(n)}, (D_2^{(n)})^\top G_{\alpha}^{(n)} D_2^{(n)}$\;
% $\eta_{\min}^{(m)}\leftarrow$ the smallest
% eigenvalue of $A_2^{(m)}$\;
$A^{(n)}\leftarrow A_1^{(n)}+\lambda A_2^{(n)}$\;
$\mathbf{b}^{(n)}\leftarrow$ the unit eigenvector 
corresponding to the smallest eigenvalue of $A^{(n)}$\;
$e_1^{(n)}, e_2^{(n)}\leftarrow
(\mathbf{b}^{(n)})^\top A_1^{(n)}\mathbf{b}^{(n)},
(\mathbf{b}^{(n)})^\top A_2^{(n)}\mathbf{b}^{(n)}$\;
  \eIf{$n=n_{\max}$}{
  $\mathbf{b}_{WGM}\leftarrow\mathbf{b}^{(n)}$\;
  \Return{} %\Comment*[r]{This is a comment}
  }{\If{$e_1^{(n)}\le\epsilon_{AD}$
and $|e_2^{(n)}-e_2^{(n-1)}|
\le \epsilon_{WG}$}{
  $\mathbf{b}_{WGM}\leftarrow\mathbf{b}^{(n)}$\;
  \Return{}
    }
  }
$n\leftarrow n+1$\;
}
\end{algorithm}

The adaptive implicitization methodology 
for polynomial curves, called the Weak Gradient Method (WGM for short), is summarized in Algorithm 1.
The inputs of Algorithm 1 are the parametric curve
$\mathbf{p}(t)$, the maximum implicit degree
$n_{\max}$, and thresholds $\epsilon_{AD}, \epsilon_{WG}$. 
The line 1 is employed to initialize the implicit
degree $n$. In the while loop (line 2 to line 20),
the matrix $A^{(n)}$ in the objective function $L_{\lambda,n}(\mathbf{b})$
is computed first using the collocation matrix 
$D_1^{(n)}$, $D_2^{(n)}$ and the Gram matrix $G_{\alpha}^{(n)}$;
then for the $n$th (current) cycle,
the ``optimal'' coefficient vector $b^{(n)}$ is found, 
the AD error 
$e_{1}^{(n)}$ and the WG error 
$e_{2}^{(n)}$ are computed subsequently;
if $n=n_{\max}$, which means that
the coefficient vectors
$b^{(1)},\ldots,b^{(n-1)}$ obtained 
in previous cycles are unacceptable,
then $b^{(n)}$ is treated as the final result and 
the algorithm is terminated under this 
circumstance (line 10 to line 12);
if $n<n_{\max}$ and $e_{1}^{(n)}, e_{2}^{(n)}$
satisfy the stopping criterions (line 14),
then return $b^{(n)}$ as the final result and 
terminate the algorithm; 
if none of the aforesaid If statement holds,
then let $n$ increases by one and go to the next cycle.

\subsection{Experiments}
\begin{figure}[h]%
\centering
\includegraphics[width=0.5\textwidth]{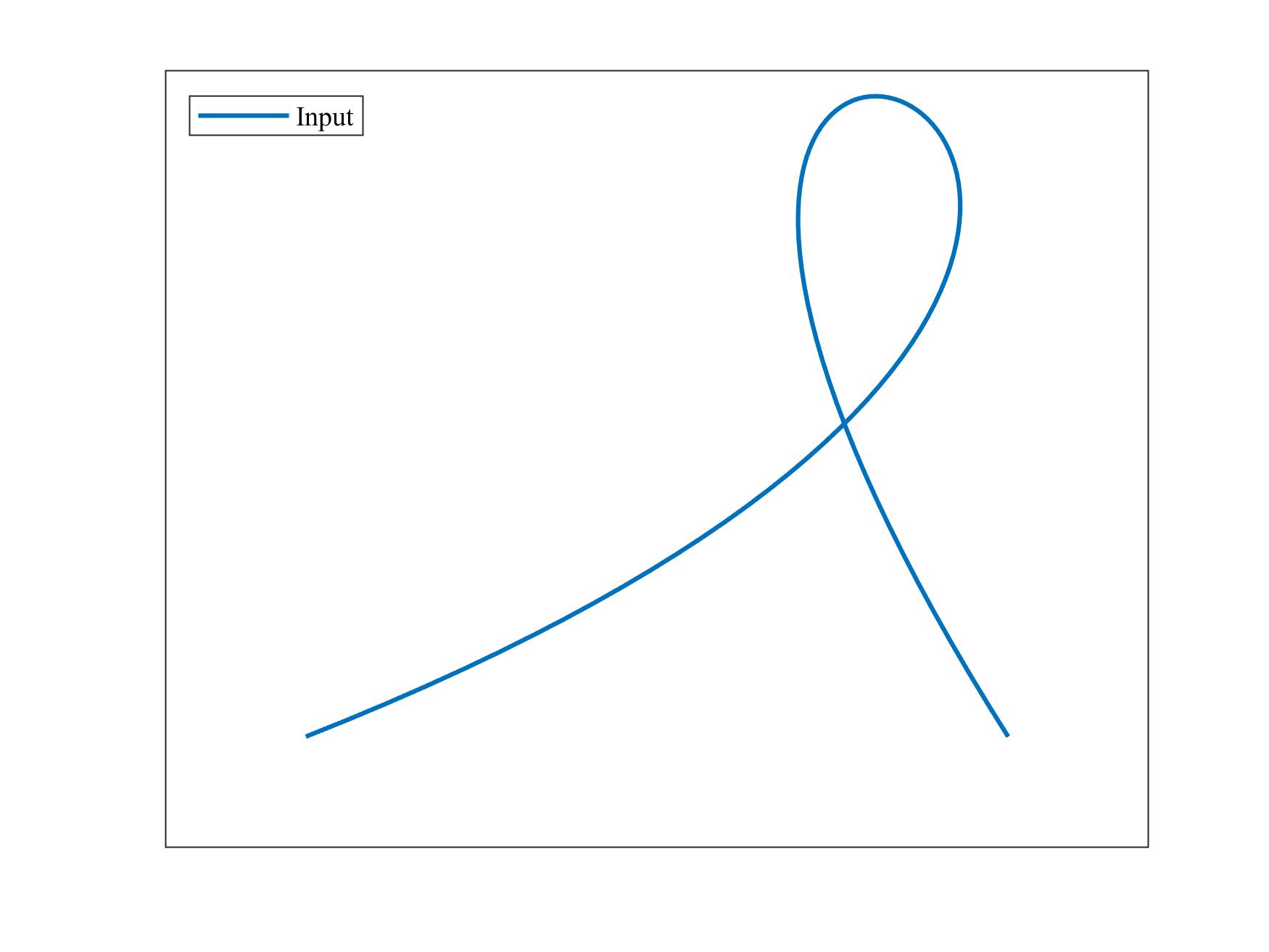}
\caption{The input curve $C_1(t)$.}\label{input_ex1}
\end{figure}
\begin{sidewaysfigure}
% \begin{figure}[htbp]
	\centering
	\subfigure[WGM ($n=2$)]{\begin{minipage}[c]{0.25\textwidth}
		\centering
		\includegraphics[width=\textwidth]{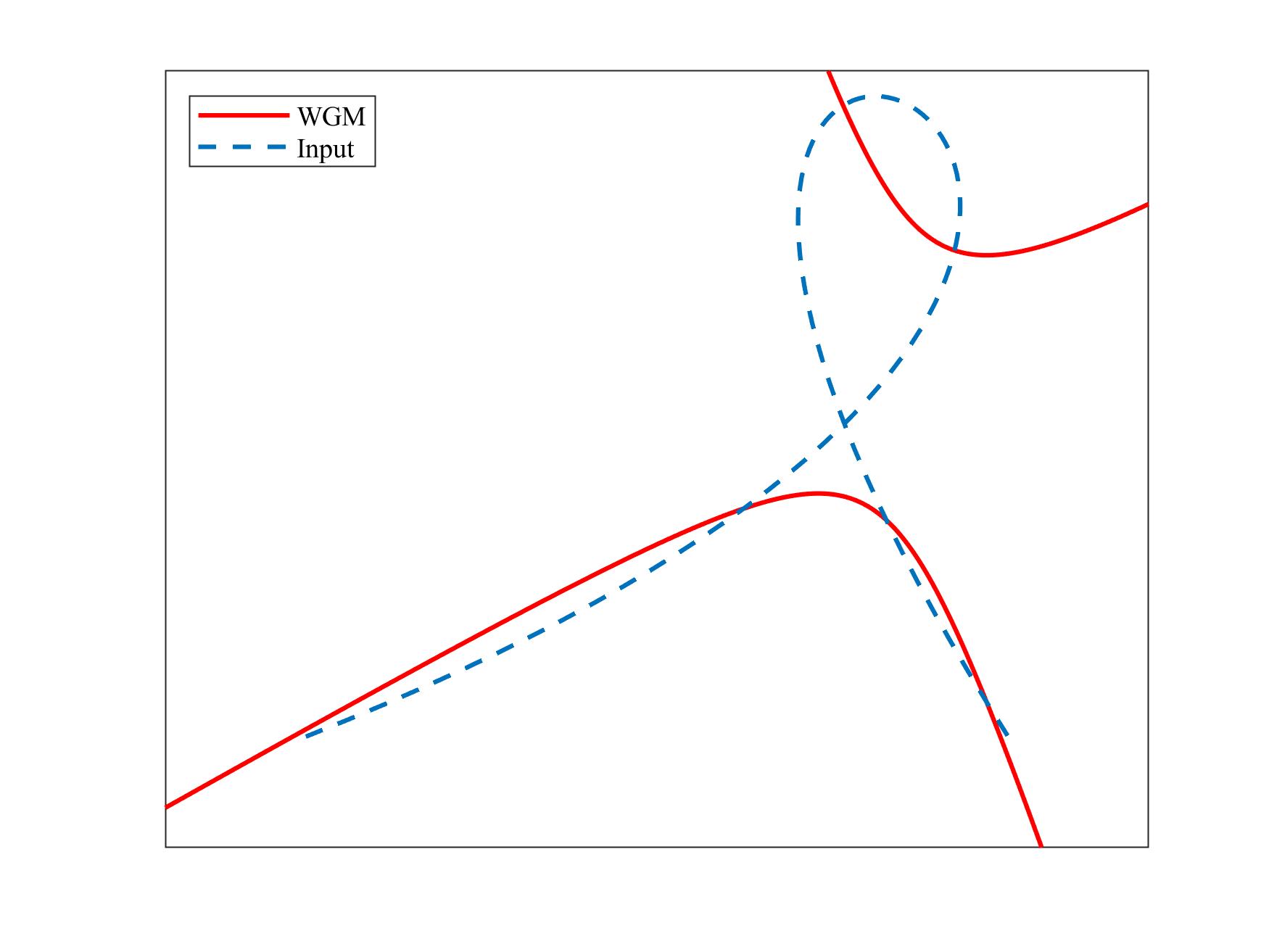}
	\end{minipage}}
	 \subfigure[WGM ($n=3$)]{\begin{minipage}[c]{0.25\textwidth}
		\centering
		\includegraphics[width=\textwidth]{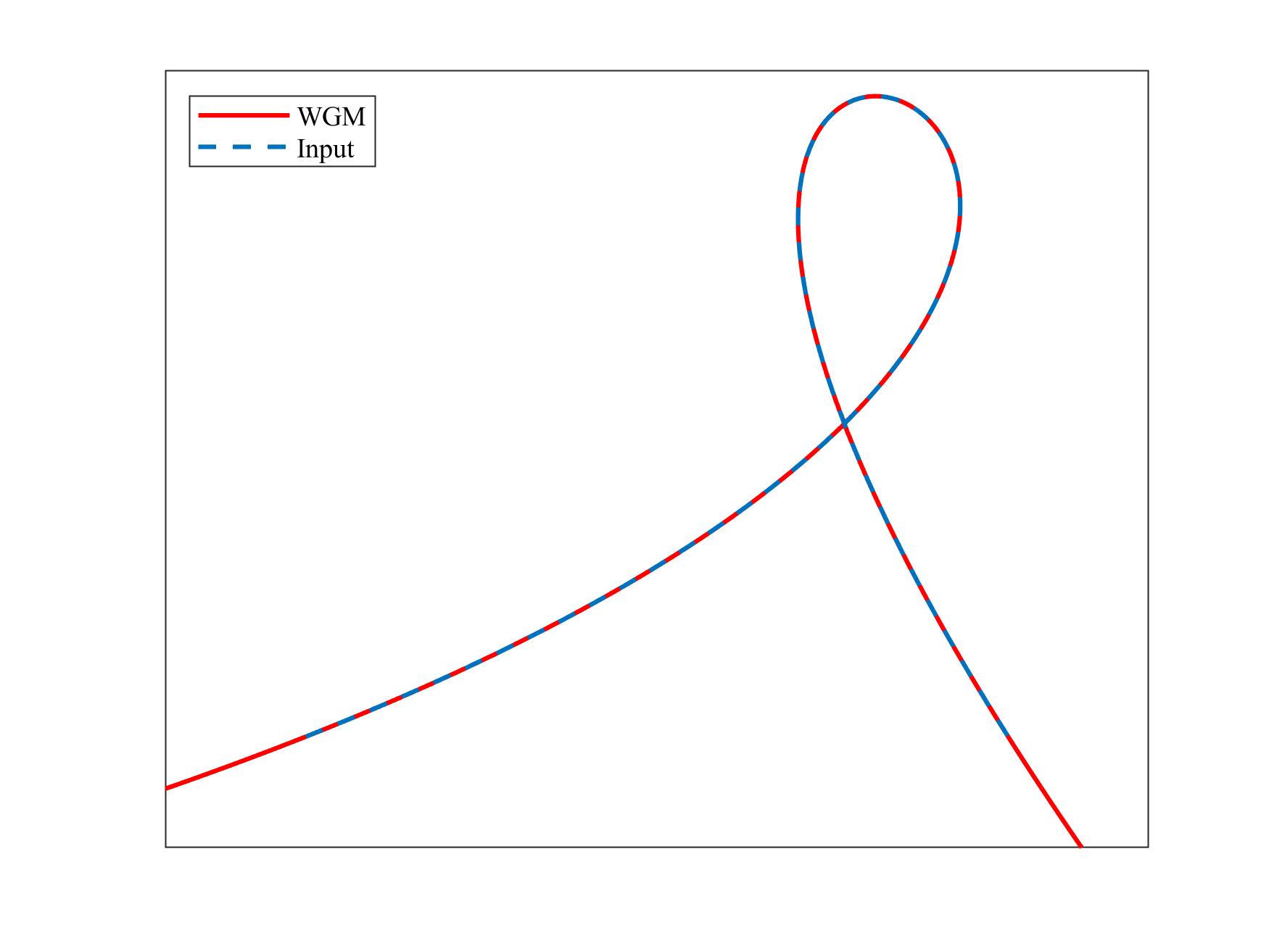}
	\end{minipage} }
	\subfigure[WGM ($n=5$)]{\begin{minipage}[c]{0.25\textwidth}
		\centering
		\includegraphics[width=\textwidth]{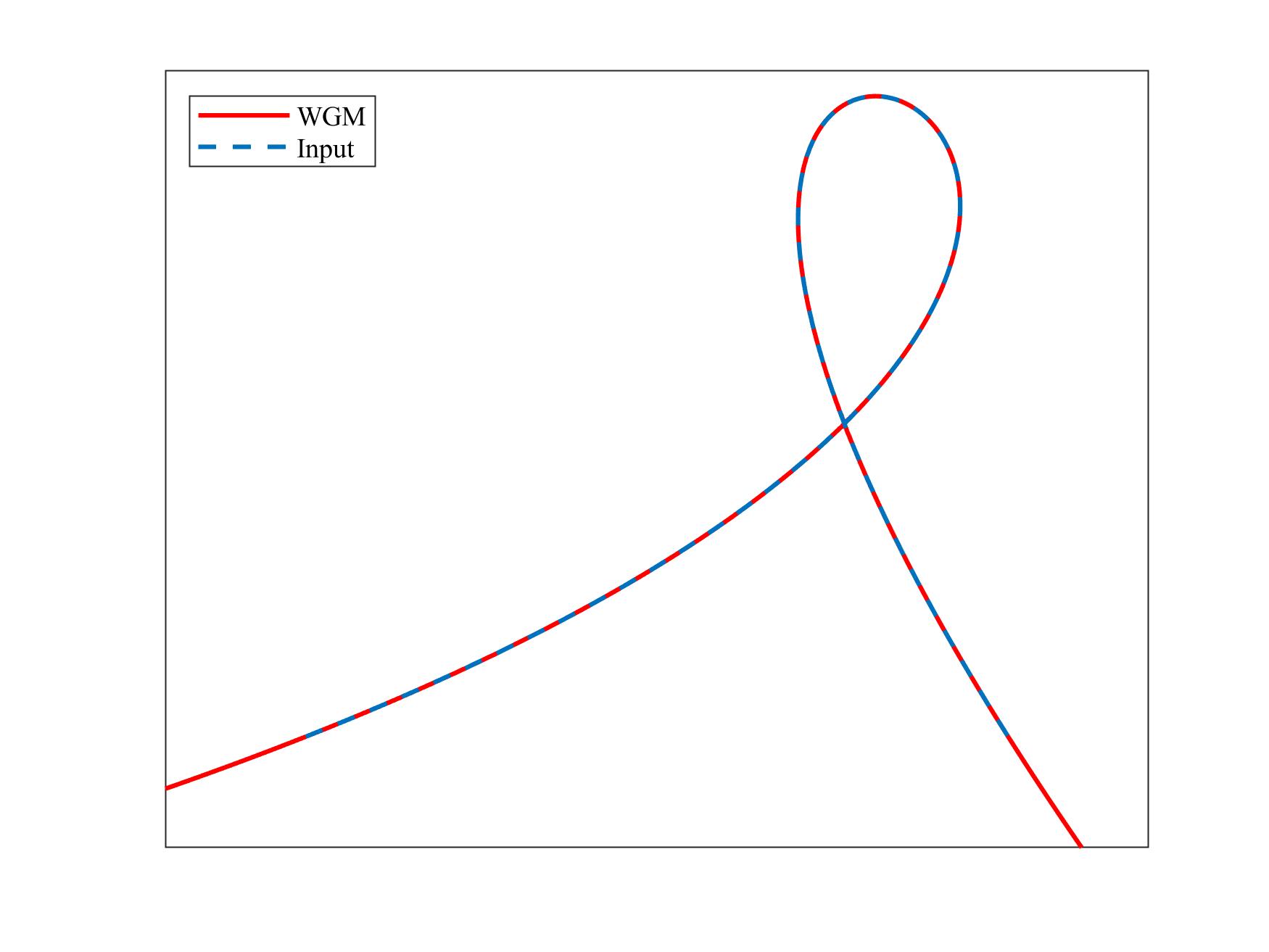}
	\end{minipage}}
	\\
	\subfigure[DM ($n=2$)]{\begin{minipage}[c]{0.25\textwidth}
		\centering
		\includegraphics[width=\textwidth]{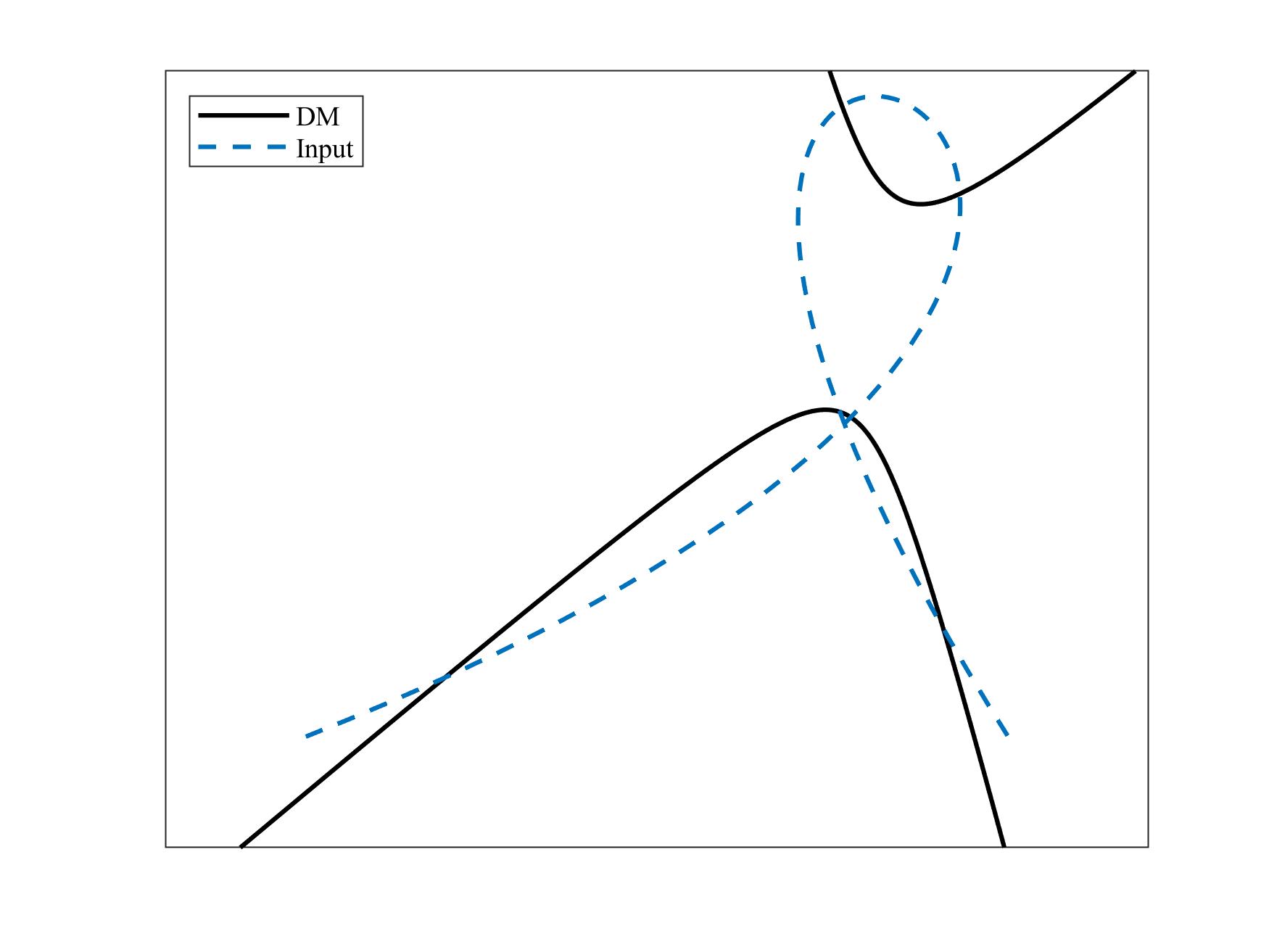}
	\end{minipage}}
	 \subfigure[DM ($n=3$)]{\begin{minipage}[c]{0.25\textwidth}
		\centering
		\includegraphics[width=\textwidth]{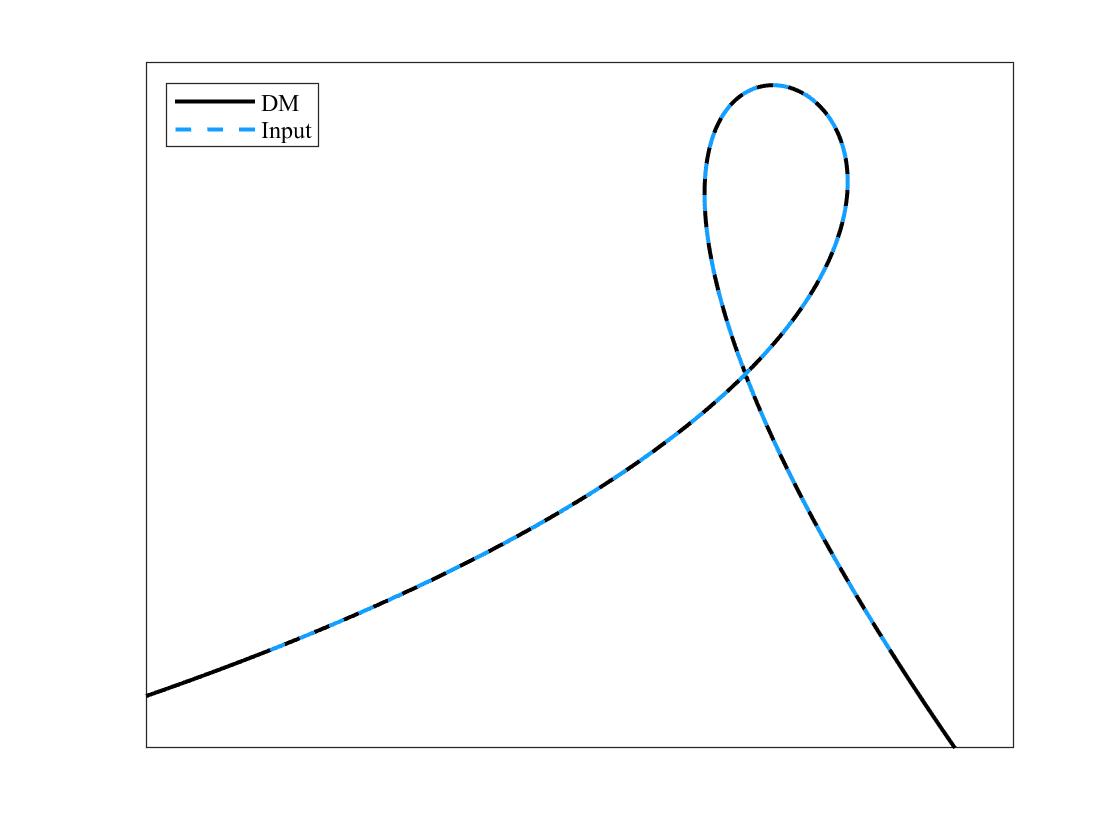}
	\end{minipage} }
	\subfigure[DM ($n=5$)]{\begin{minipage}[c]{0.25\textwidth}
		\centering
		\includegraphics[width=\textwidth]{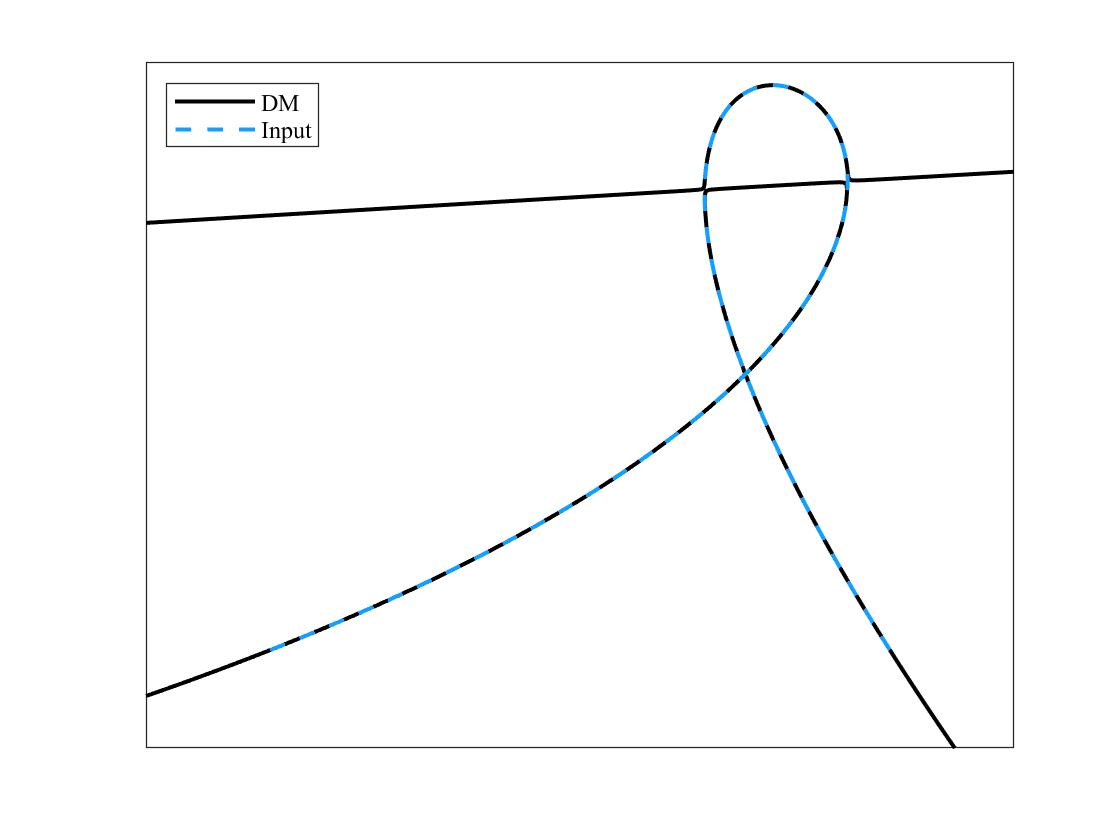}
	\end{minipage}}
	\caption{Adaptive implicitization of $C_1(t)$.
	The blue dash line in (a)-(f) is the input curve, 
	the red line in (a)-(c) is the output curve by our method, 
	and the black line in (d)-(f) is the output curve by Dokken's method. 
	From left to right: the implicit degree $n=2,3,5$.}
	\label{ex1_output}
% \end{figure}
\end{sidewaysfigure}

\begin{figure}[h]%
\centering
\includegraphics[width=0.5\textwidth]{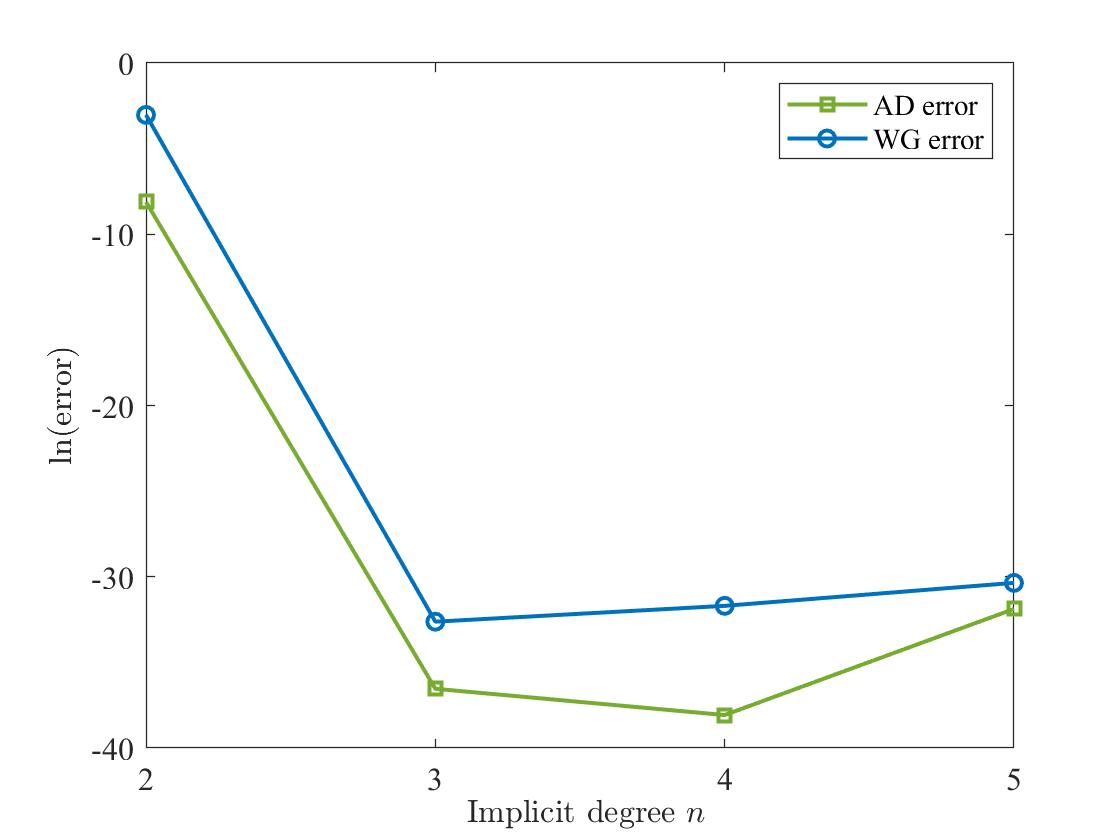}
\caption{Statistics of our method on changes of the AD and WG error for $C_1(t)$, as the implicit degree $n$ increases.}\label{error1}
\end{figure}

Additional branches (i.e. the extra zero contour) generated in the implicitization procedure
make the resulting curves challenging to be interpreted, and the elimination of additional branches 
is the main problem in designing the implicitization methods. 
\cite{juttler2005approximate} address this problem
by combining two or more eigenvectors (associated with small eigenvalues), 
which leads to a gradual decline of the accuracy.
With the WGM developed in this paper, 
additional branches can be avoided as much as possible in the implicitization procedure. 

We choose the basis $\alpha(t)$ to be the univariate Bernstein polynomial basis
of degree $p=mn+1$, i.e.
\begin{equation*}
    \alpha_{i}(t)=B_{i,p}(t)=\frac{p!}{i!(p-i)!}t^i(1-t)^{p-i},i=0,1,\ldots,p,
\end{equation*}
and the basis $\left \{ \phi_i(x,y) \right \}_{i=1}^{k}$ to be 
the bivariate monomial basis of total degree $n$.
We set the maximum implicit degree $n_{\max}=7$, the regulator gain $\lambda=0.1$, and the thresholds $\epsilon_{AD}=10^{-4}$ and $\epsilon_{WG}=10^{-3}$.
\begin{example}
Consider the polynomial parametric curve
\begin{equation*}
    C_1(t)=\left ( \begin{matrix}0
 \\
0
\end{matrix} \right ) B_{0,3}(t)+
\left ( \begin{matrix}2
 \\
1
\end{matrix} \right ) B_{1,3}(t)+
\left ( \begin{matrix}0
 \\
2
\end{matrix} \right )  B_{2,3}(t)+
\left ( \begin{matrix}1
 \\
0
\end{matrix} \right ) B_{3,3}(t),
\end{equation*}
where the parameter of $C_1(t)$ takes value in 
$[0, 1]$, and 
$\left\{ B_{i,3}(t)\right\}_{i=0}^3$ 
is the Bernstein polynomial basis of degree $3$.
$C_1(t)$ is shown in Figure \ref{input_ex1}.

The first row of Figure \ref{ex1_output} shows the adaptive implicitization process of $C_1(t)$ by the WGM.
Similarly, the second row of Figure \ref{ex1_output} shows the implicitization process of $C_1(t)$ by 
Dokken’s method. We can see that for every iteration, the WGM refrains from additional branches as much as possible, see (c) vs (f) in Figure \ref{ex1_output}. 

Figure \ref{error1} shows the statistic of our method on changes of the AD and WG error for $C_1(t)$, when the implicit degree $n$ is increasing.
\end{example}
\begin{figure}[h]%
\centering
\includegraphics[width=0.5\textwidth]
{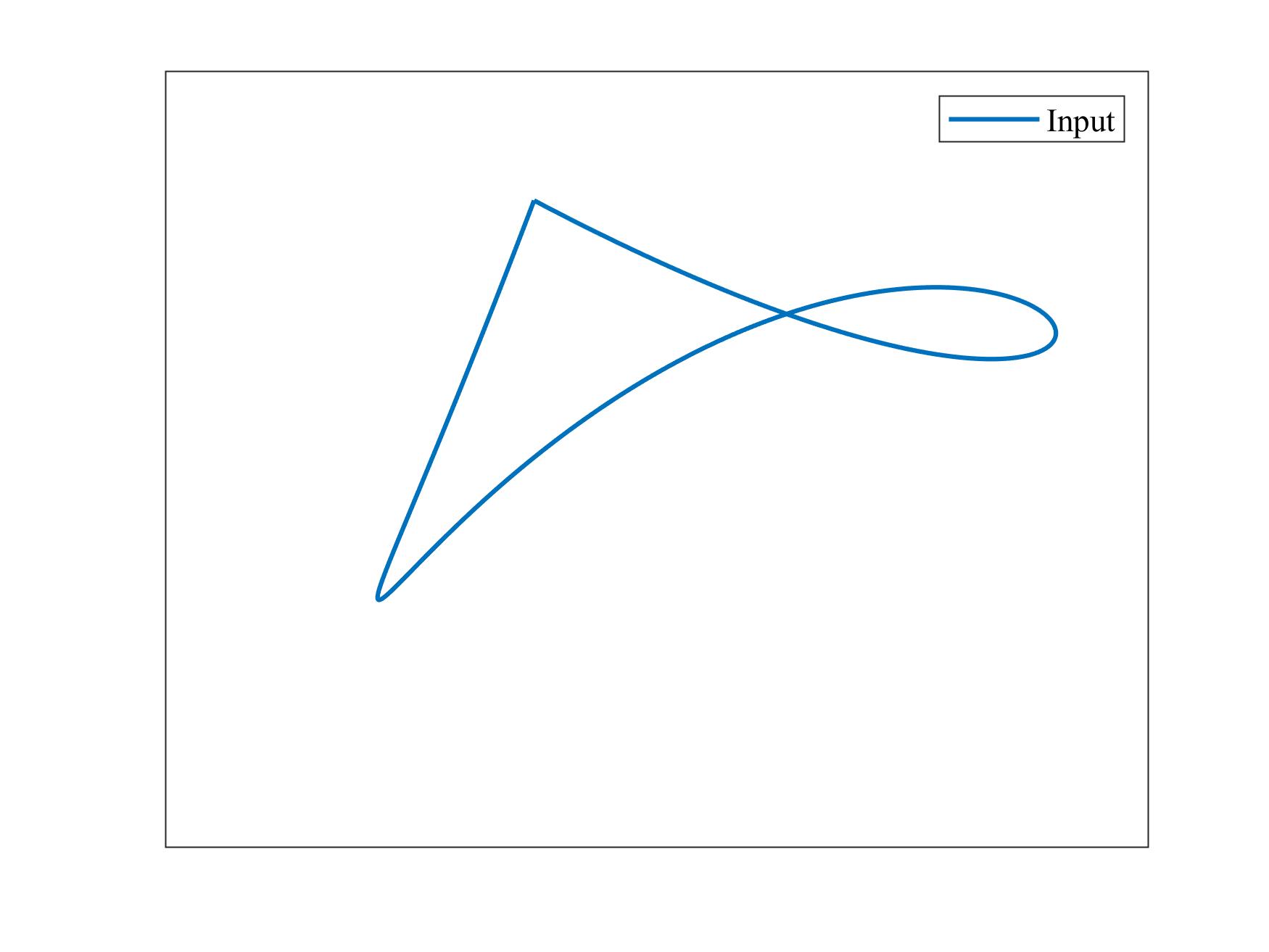}
\caption{The input curve $C_2(t)$.}\label{input_ex2}
\end{figure}
\begin{example}
Consider the polynomial parametric curve
\begin{equation*}
    C_2(t)=\left ( \begin{matrix}1
 \\
5
\end{matrix} \right ) B_{0,4}(t)+
\left ( \begin{matrix}-3
 \\
-15
\end{matrix} \right ) B_{1,4}(t)+
\left ( \begin{matrix}2
 \\
20
\end{matrix} \right )  B_{2,4}(t)+
\left ( \begin{matrix}11
 \\
-5
\end{matrix} \right ) B_{3,4}(t) +
\left ( \begin{matrix}1
 \\
5
\end{matrix} \right ) B_{4,4}(t),
\end{equation*}
where the parameter of $C_2(t)$ take values in 
$[0, 1]$, and 
$\left\{ B_{i,4}(t)\right\}_{i=0}^4$ 
is the Bernstein polynomial basis of degree $4$.
$C_2(t)$ is shown in Figure \ref{input_ex2}.

The first row of Figure \ref{ex2_output} shows the adaptive implicitization process of $C_2(t)$ by the WGM.
Similarly, the second row of Figure \ref{ex2_output} shows the implicitization process of $C_2(t)$ 
by Dokken’s method. We can see that for every iteration, 
the WGM's output curve will approach $C_2(t)$
closer than that of Dokken's method, from the viewpoint of ``shape-preserving''.
Moreover, the WGM refrains from additional branches as much as possible, 
see (c) vs (h), (d) vs (i), and (e) vs (j) in Figure \ref{ex2_output}.
\begin{sidewaysfigure}
	\subfigure[WGM ($n=2$)]{\begin{minipage}[c]{0.19\textwidth}
		\centering
		\includegraphics[width=\textwidth]{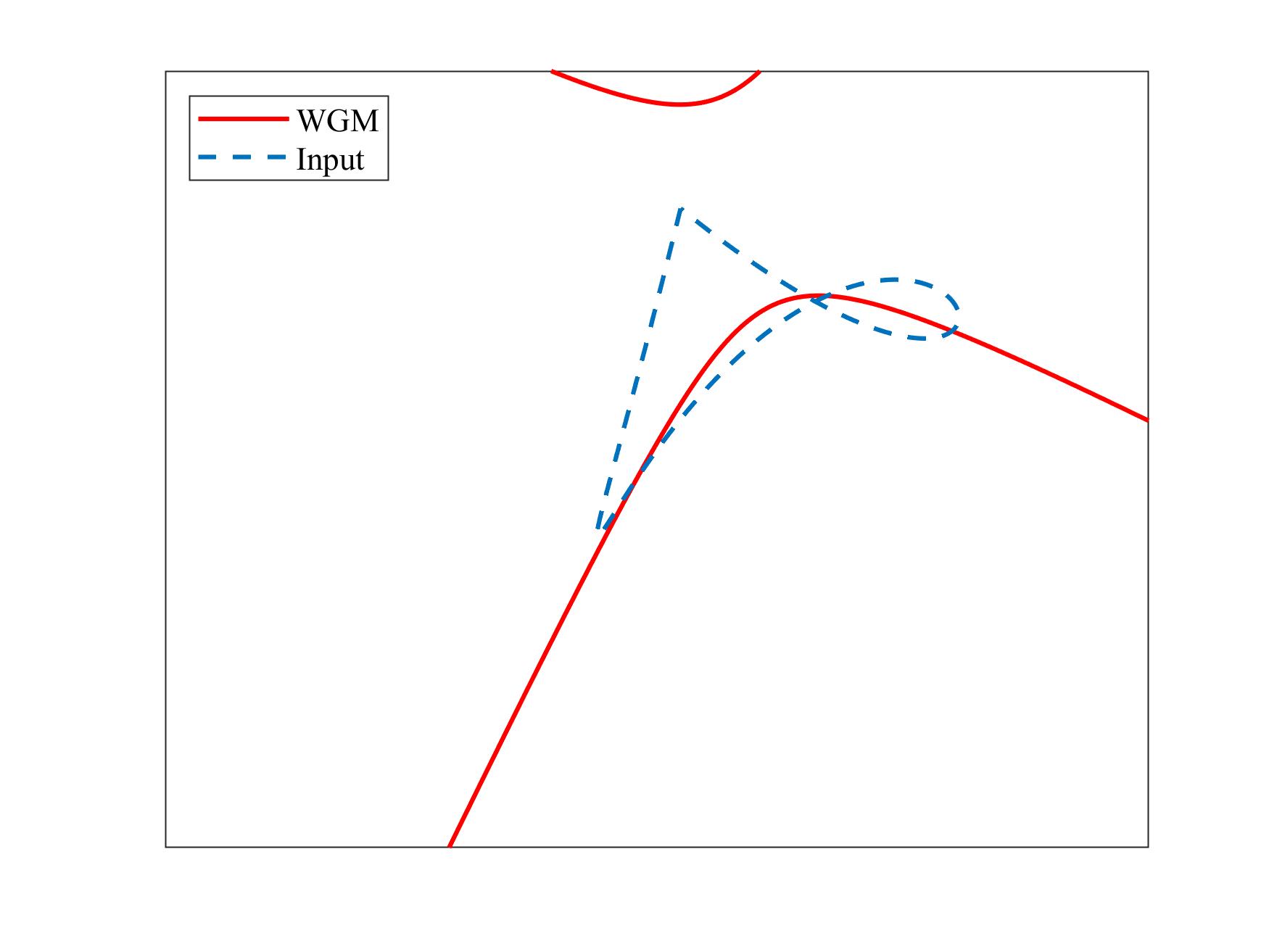}
	\end{minipage}}
	\subfigure[WGM ($n=3$)]{\begin{minipage}[c]{0.19\textwidth}
		\centering
		\includegraphics[width=\textwidth]{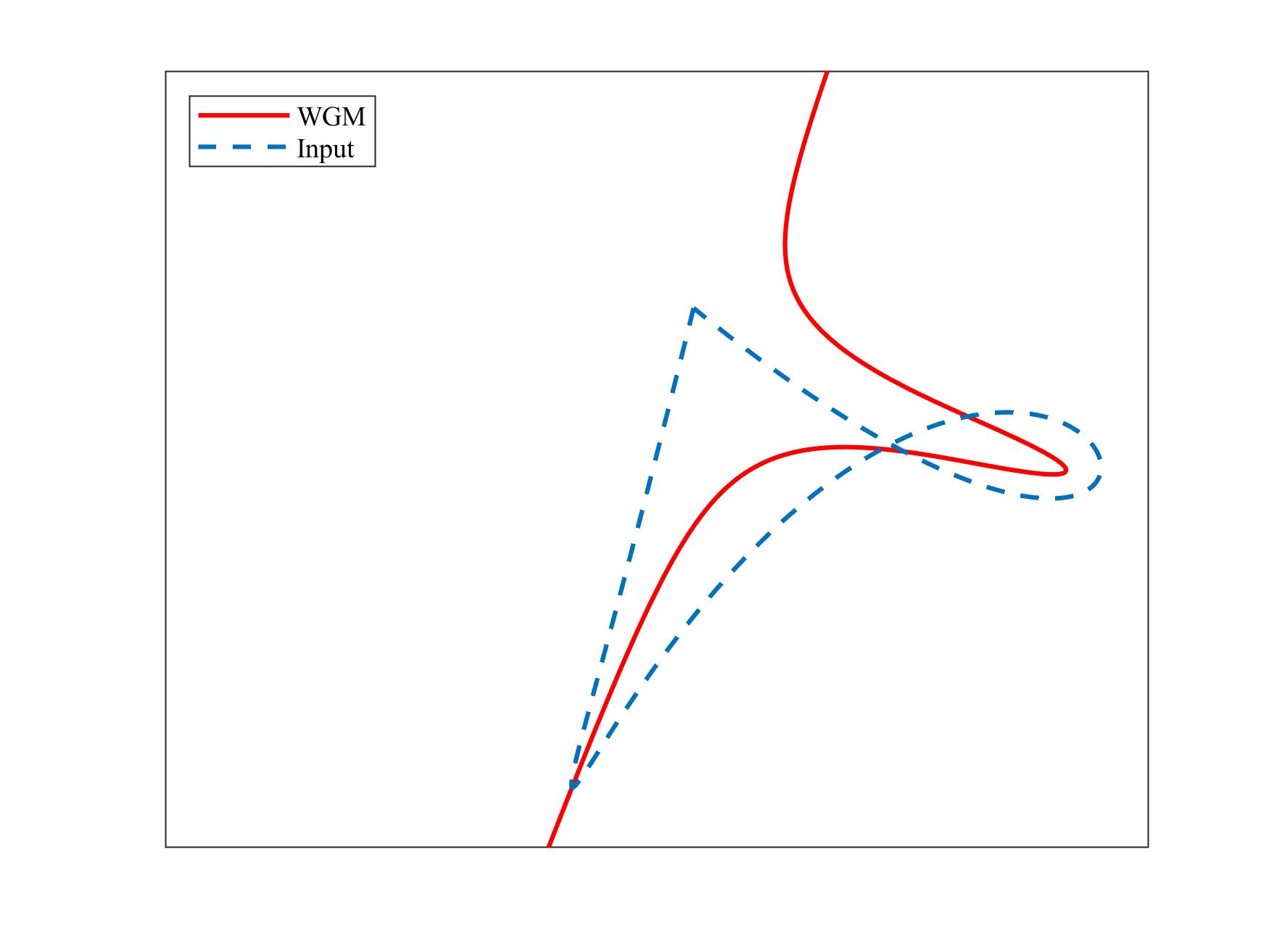}
	\end{minipage}}
	\subfigure[WGM ($n=4$)]{\begin{minipage}[c]{0.19\textwidth}
		\centering
		\includegraphics[width=\textwidth]{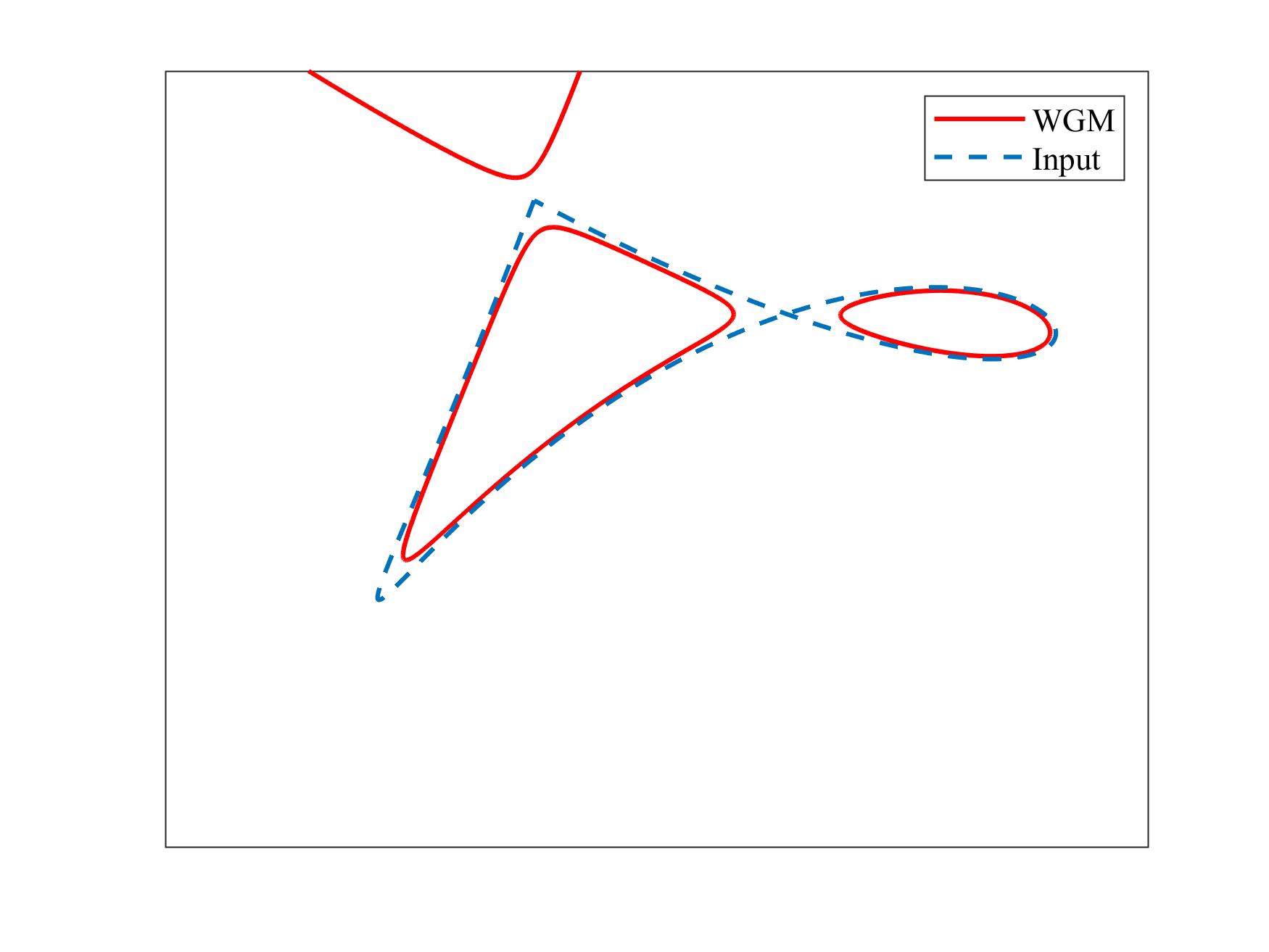}
	\end{minipage}}
	\subfigure[WGM ($n=5$)]{\begin{minipage}[c]{0.19\textwidth}
		\centering
		\includegraphics[width=\textwidth]{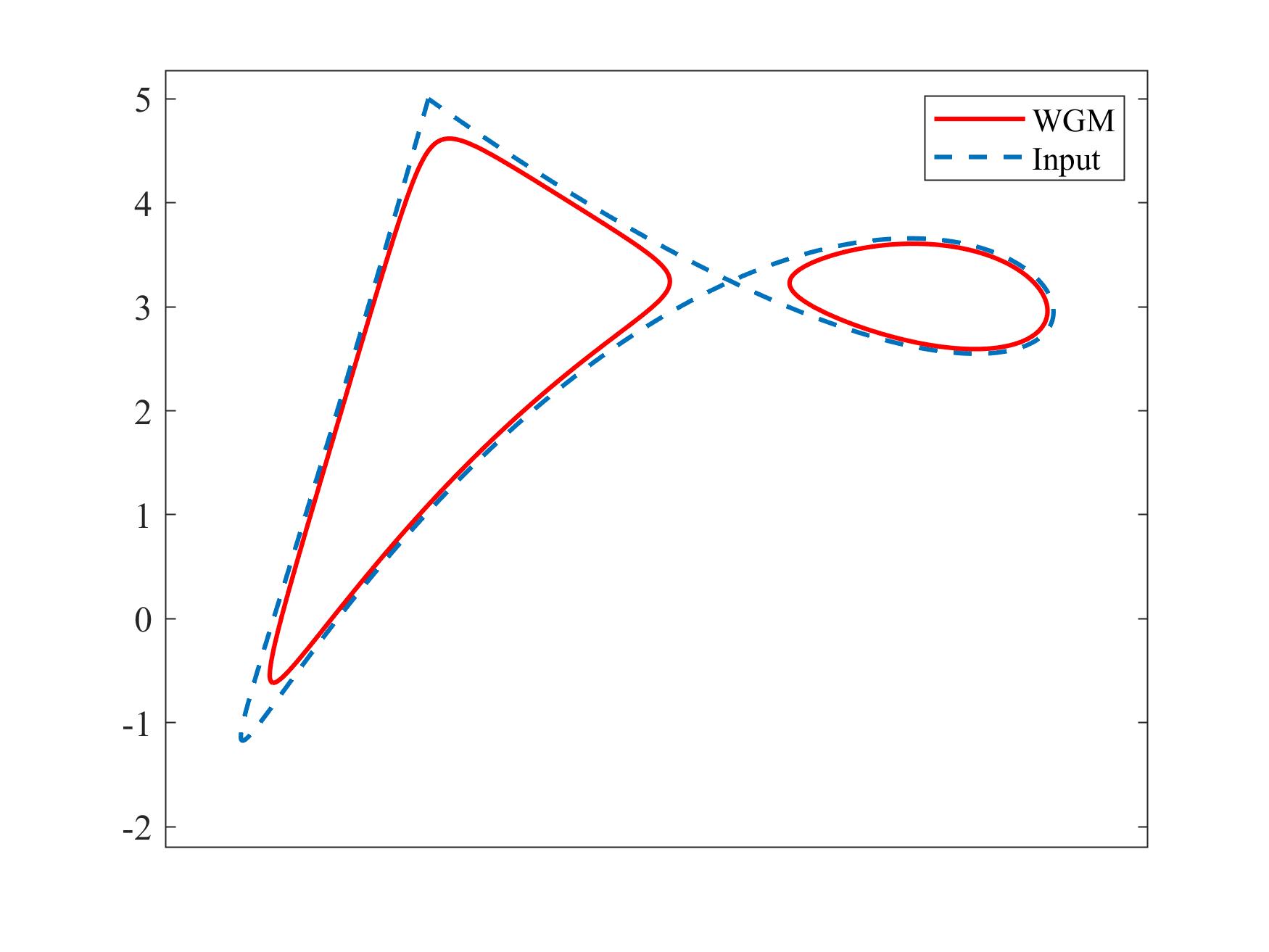}
	\end{minipage}}
	\subfigure[WGM ($n=6$)]{\begin{minipage}[c]{0.19\textwidth}
		\centering
		\includegraphics[width=\textwidth]{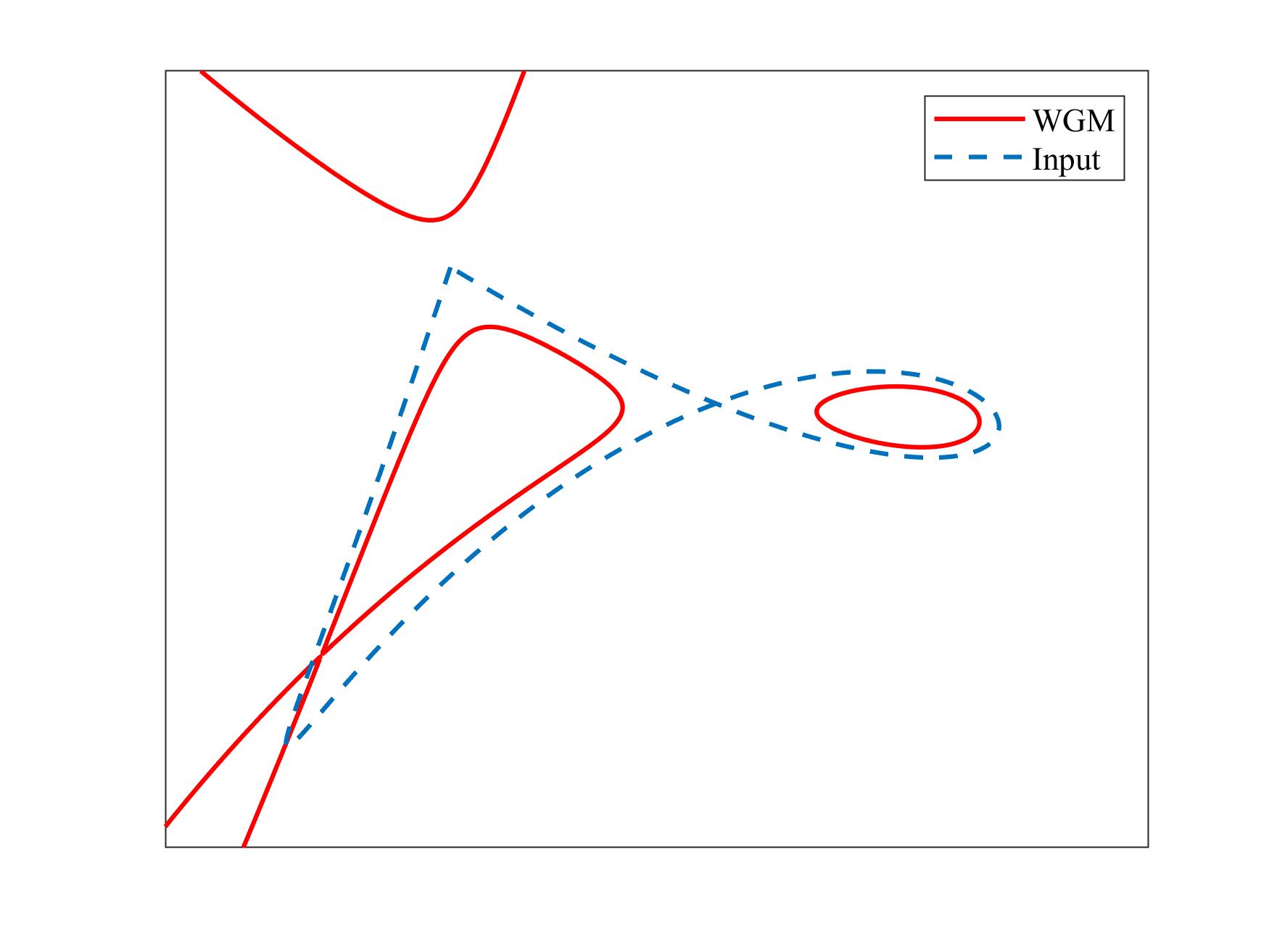}
	\end{minipage}}
	\\
	\subfigure[DM ($n=2$)]{\begin{minipage}[c]{0.19\textwidth}
		\centering
		\includegraphics[width=\textwidth]{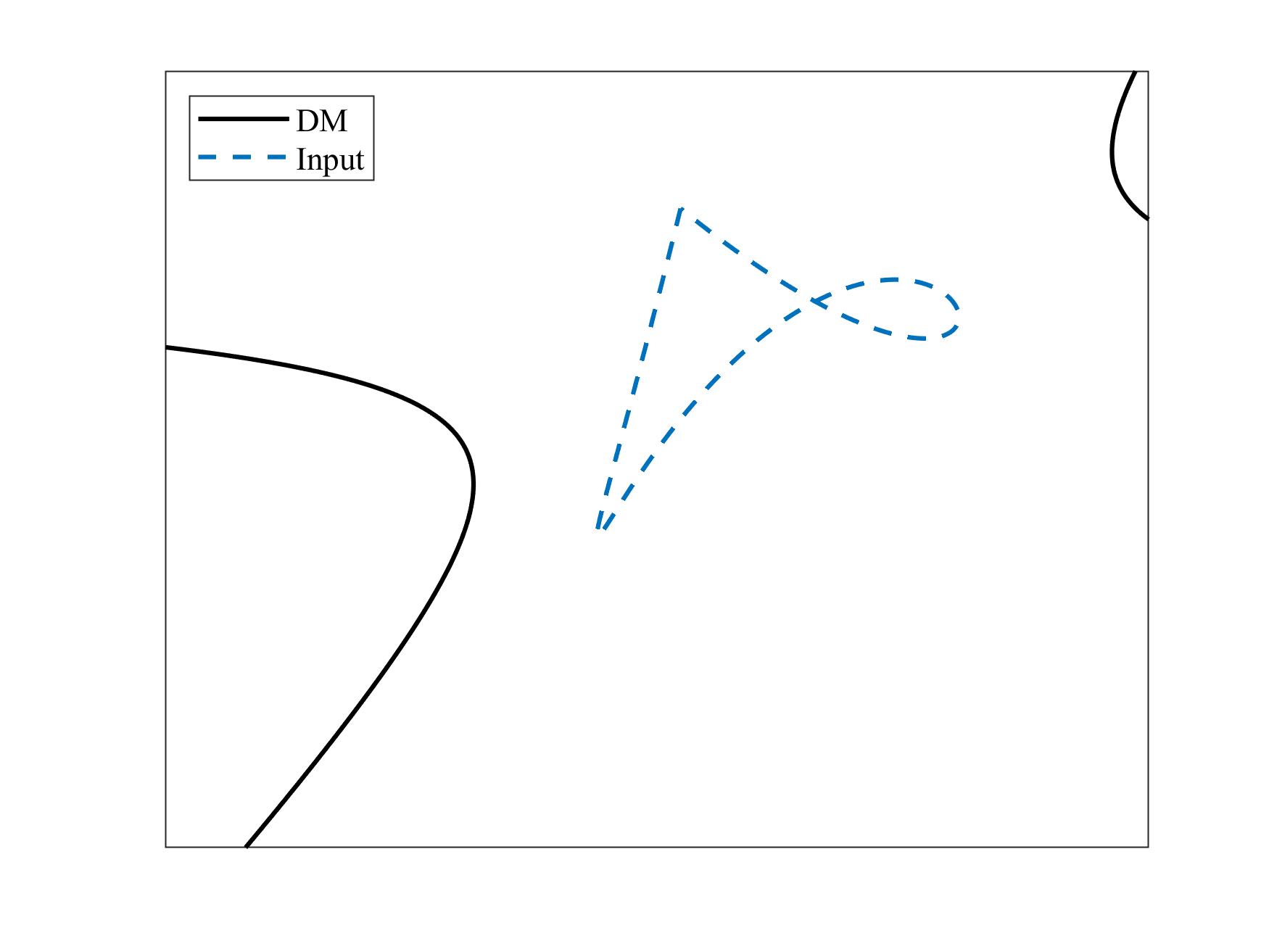}
	\end{minipage}}
	\subfigure[DM ($n=3$)]{\begin{minipage}[c]{0.19\textwidth}
		\centering
		\includegraphics[width=\textwidth]{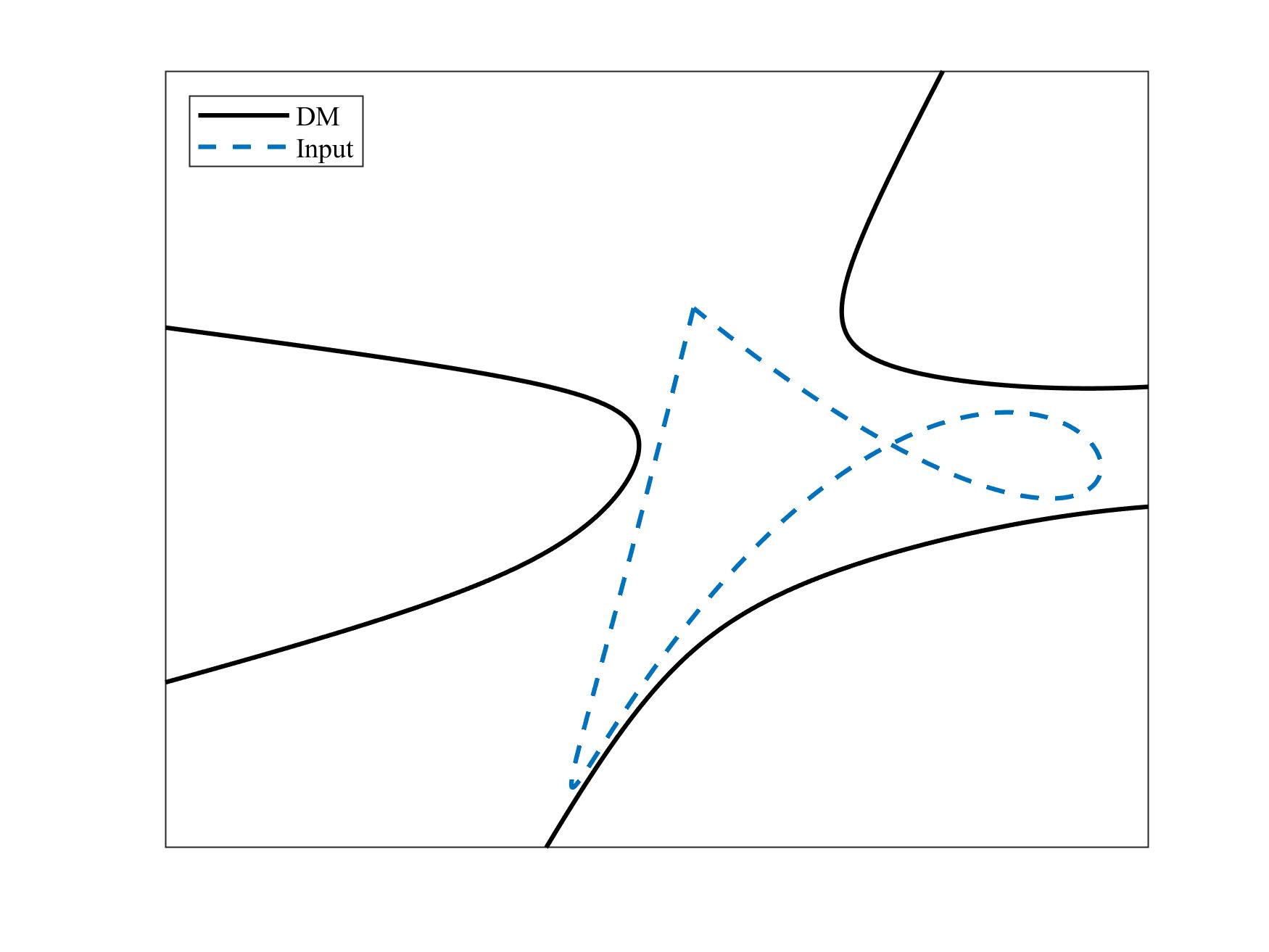}
	\end{minipage}}
	\subfigure[DM ($n=4$)]{\begin{minipage}[c]{0.19\textwidth}
		\centering
		\includegraphics[width=\textwidth]{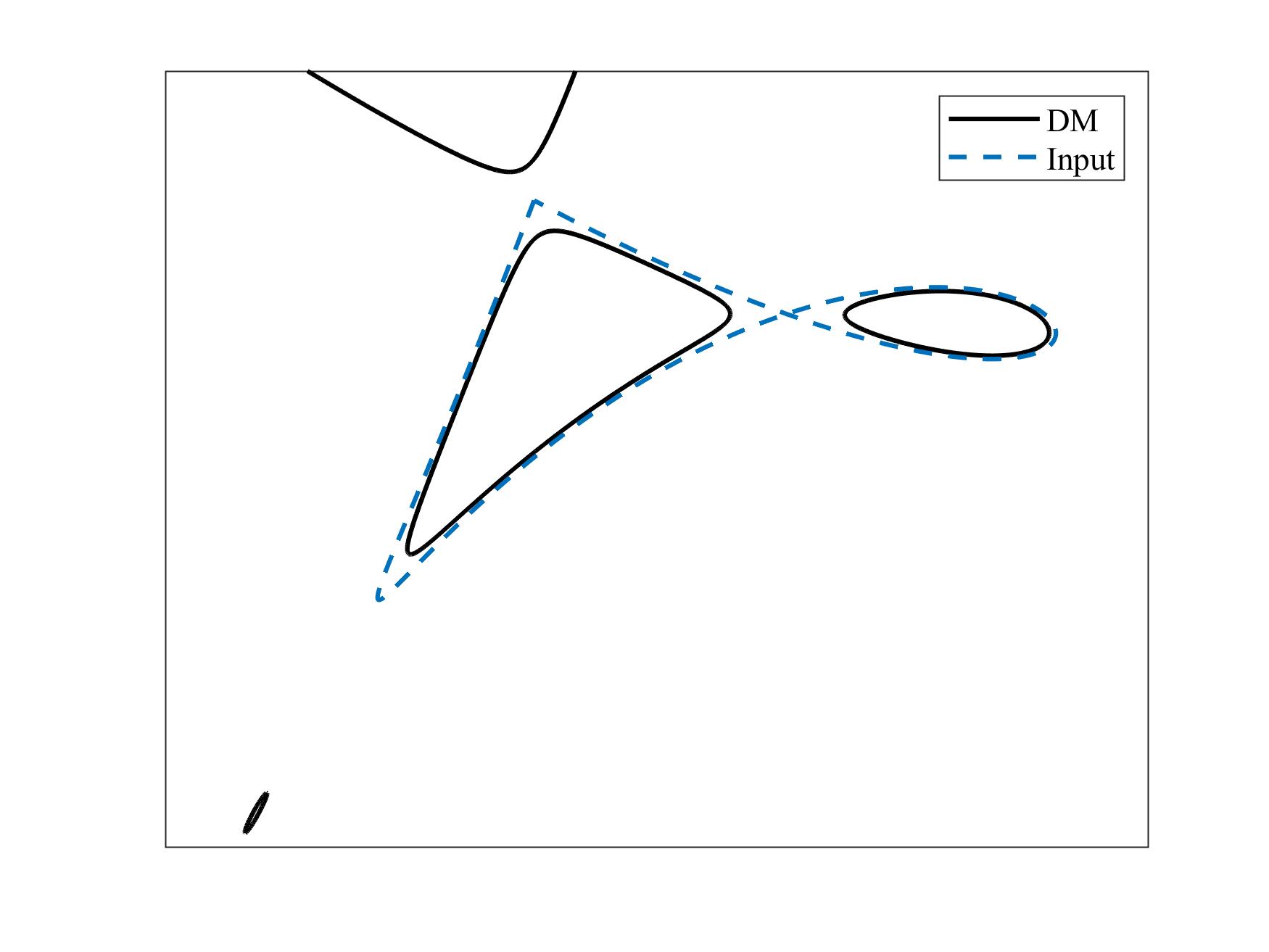}
	\end{minipage}}
	\subfigure[DM ($n=5$)]{\begin{minipage}[c]{0.19\textwidth}
		\centering
		\includegraphics[width=\textwidth]{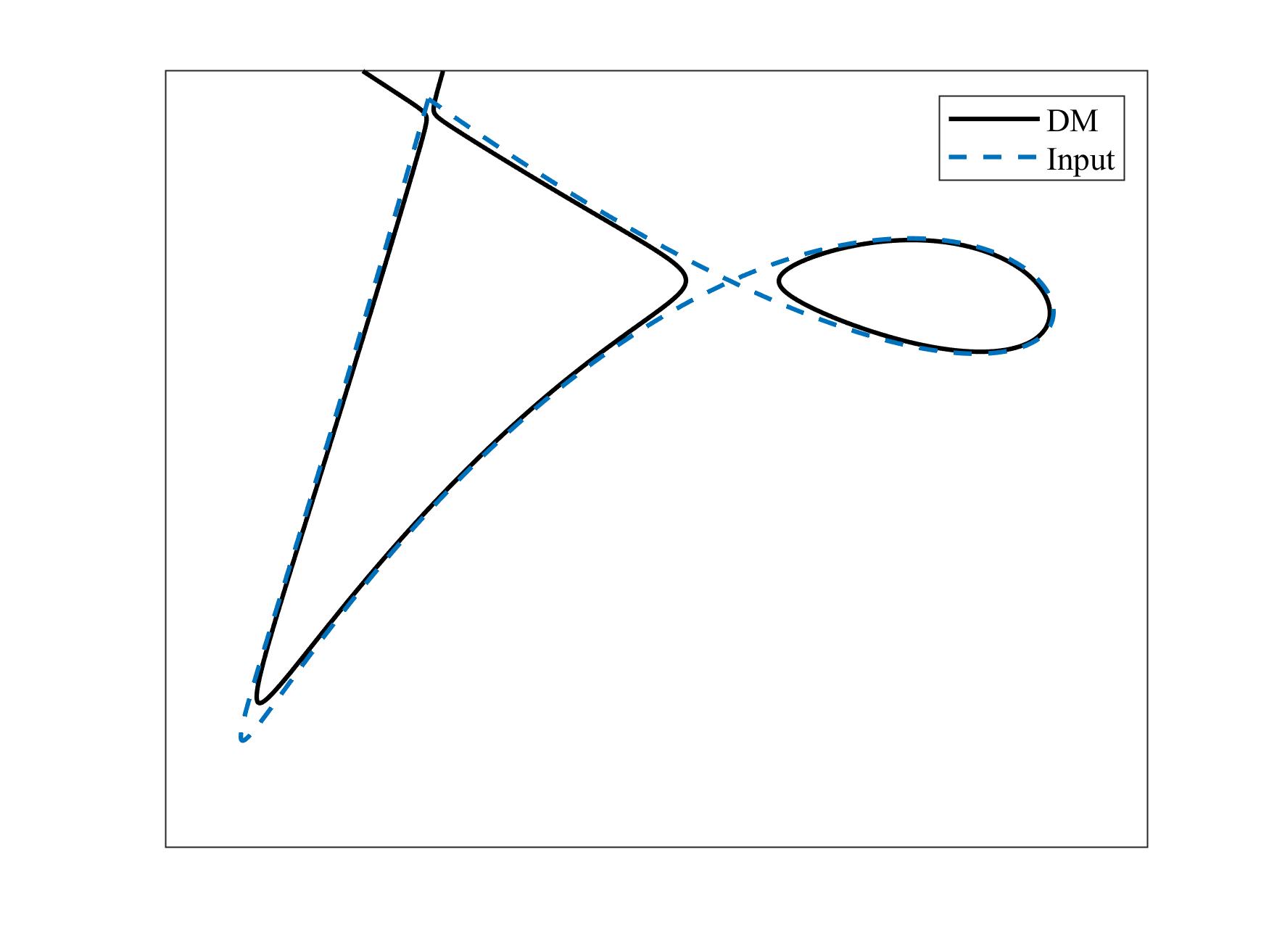}
	\end{minipage}}
	\subfigure[DM ($n=6$)]{\begin{minipage}[c]{0.19\textwidth}
		\centering
		\includegraphics[width=\textwidth]{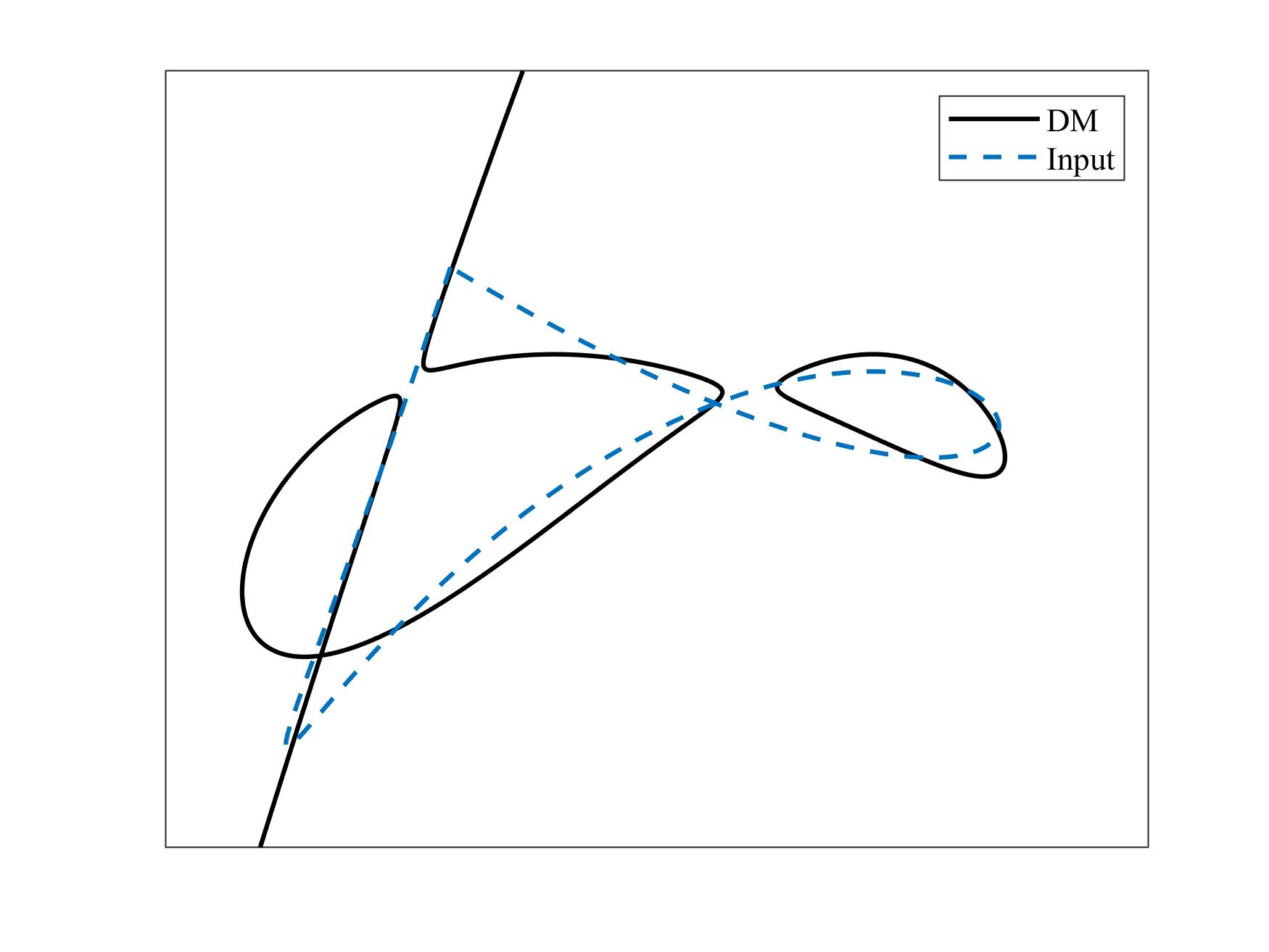}
	\end{minipage}}
	\caption{Adaptive implicitization of $C_2(t)$.
	The blue dash line in (a)-(j) is the input curve, 
	the red line in (a)-(e) is the output curve by our method, 
	and the black line in (f)-(i) is the output curve by Dokken's method. 
	From left to right: the implicit degree $n=2,3,4,5,6$.}
	\label{ex2_output}
\end{sidewaysfigure}
\begin{figure}[h]%
    \centering
    \includegraphics[width=0.5\textwidth]{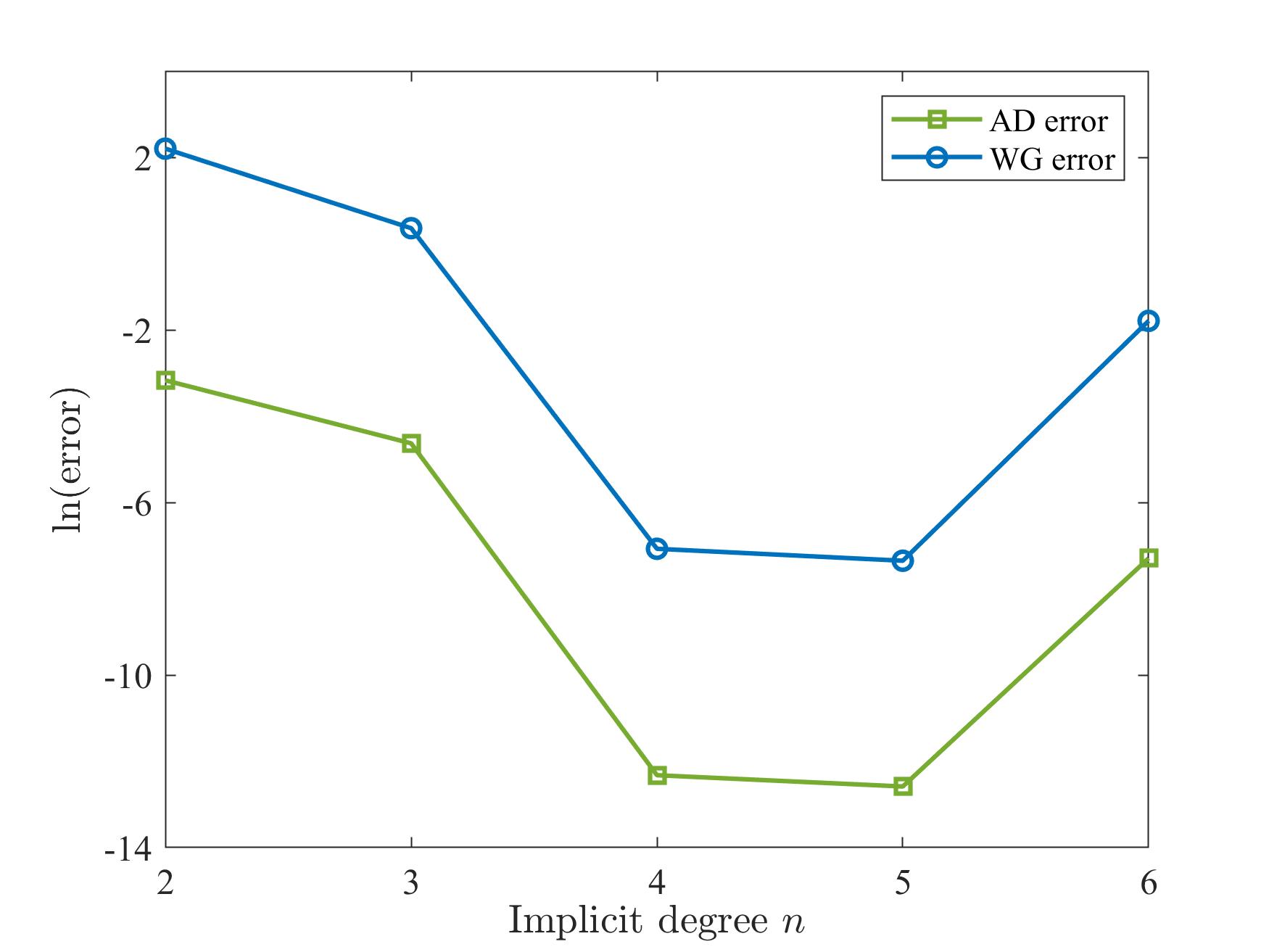}
    \caption{Statistics of our method on changes of the AD and WG error for $C_2(t)$,
    as the implicit degree $n$ increases.}
    \label{error2}
\end{figure} 

Figure \ref{error2} shows the statistic of our method on changes 
of the AD and WG error for $C_2(t)$, when the implicit degree $n$ is increasing.
\end{example}

Finally, Table \ref{tab1} shows the comparison of running time performance of WGM, DM 
and the method in \cite{interian2017curve}.
\begin{table}[h!]
  \begin{center}
    \caption{Performance of WGM on $C_1(t), C_2(t)$ (all timings are measured in seconds).}
    \label{tab1}
    \begin{tabular}{cccccc} 
          \hline
          \hline
Input	& AD Error	& WG Error 
& WGM & DM & Method in \cite{interian2017curve}\\
 	&  	&   
& Time  & Time  & Time \\
      \hline
$C_1(t)$		& $2.312e-16$			& $1.266e-14$  & $0.0019$  & $0.0016$ & $48.71$\\
$C_2(t)$		& $3.727e-6$			& $9.119e-4$  & $0.0030$  & $0.0023$ & $128.81$\\
          \hline
    \end{tabular}
  \end{center}
\end{table}
\section{Methodology for Non-Polynomial Curves}\label{non polynomial method}
To deal with non-polynomial curves, our approach in this section
is to sample a number of points with associated oriented tangent vectors, 
and then to convert them into implicit polynomials by 
discrete approximate implicitization (i.e. implicit fitting).
\subsection{Curve sampling}
Here, we employ the 
uniform sampling, which makes the sample points distribute
uniformly in the parametric space.
% High quality sampling results can improve the quality of implicitization to parametric curves.
% We adopt a hybrid sampling method for the input parametric curves ${\mathbf{p}(t)}$.
% In addition to points for equally spaced parameters, we also pick two kinds of feature points:
% \begin{itemize}
%     \item inflection point, at which the curvature vanishes and changes sign 
%     in the neighborhood of the point;
%     \item extreme curvature point, at which the derivative of curvature changes sign.
% \end{itemize}
All the sampling points and their associated tangent vectors are respectively denoted by 
$\left \{ \mathbf{p}_j \right \}_{j=1}^{N}$ and $\left \{ \mathbf{T}_j \right \}_{j=1}^{N}$.

\subsection{Discrete approximate implicitization}
Discrete approximate implicitization, also called implicit fitting, 
is to retrieve the implicit polynomial
\begin{equation*}
    f_{\mathbf{b}}(x,y) =
    \sum_{i=1}^k b_{i}\phi_i(x,y)
\end{equation*}
from the sampling points $\left \{ \mathbf{p}_i \right \}_{i=1}^{N}$ 
% with associated tangent vectors $\left \{ \mathbf{T}_i \right \}_{i=1}^{N}$, 
by minimizing the sum of the squared algebraic distances:
\begin{equation}\label{lad_np}
    L_{AD}=\sum_{j=1}^{N} (f_{\mathbf{b}}(\mathbf{p}_j))^2=\left \| \left (f_{\mathbf{b}}(\mathbf{p}_1), f_{\mathbf{b}}(\mathbf{p}_2), \ldots 
,f_{\mathbf{b}}(\mathbf{p}_N) \right )^\top \right \| ^2.
\end{equation}
We simplify $L_{AD}$ into the matrix form. While
\begin{equation*}
    f_{\mathbf{b}}(\mathbf{p}_j)=\left (
  \phi_1(\mathbf{p}_j), \phi_2(\mathbf{p}_j), \cdots, \phi_k(\mathbf{p}_j)\right )
\mathbf{b},~j=1,2,\ldots,N,
\end{equation*}
we have
\begin{equation*}
    \left (f_{\mathbf{b}}(\mathbf{p}_1), f_{\mathbf{b}}(\mathbf{p}_2), \ldots 
,f_{\mathbf{b}}(\mathbf{p}_N) \right )^\top =D_1\mathbf{b},
\end{equation*}
where $D_1=
\left (\phi_i(\mathbf{p}_j)\right )_{N\times k}$.
% \begin{equation}
% D_1=
% \left (\phi_j(\mathbf{p}_i)\right )_{N\times k}
% \left (\begin{matrix}
%   \phi_1(\mathbf{p}_1) & \phi_2(\mathbf{p}_1) & \cdots  & \phi_k(\mathbf{p}_1)\\
%   \phi_1(\mathbf{p}_2)& \phi_2(\mathbf{p}_2) & \cdots  & \phi_k(\mathbf{p}_2)\\
%   \vdots & \vdots  &  & \vdots\\
%   \phi_1(\mathbf{p}_N)& \phi_2(\mathbf{p}_N) & \cdots  & \phi_k(\mathbf{p}_N)\\
% \end{matrix}  \right ).
% \end{equation}
Then $L_{AD}$ in Equation (\ref{lad_np}) can be written as
\begin{equation}\label{dis_lad}
    L_{AD}=\mathbf{b}^\top A_1\mathbf{b},
\end{equation}
where 
\begin{equation*}
    A_1=D_1^\top D_1.
\end{equation*}

To avoid the trivial $\mathbf{b}=0$ for $L_{AD}$'s minimization, 
we introduce the weak gradient constraint for the discrete case:
\begin{align*}
    L_{WG}&=\sum_{j=1}^{N} 
    \left ( \bigtriangledown f_{\mathbf{b}}(\mathbf{p}_j)\cdot \mathbf{T}_j \right )^2\\
    &=\left \| \left (\bigtriangledown f_{\mathbf{b}}(\mathbf{p}_1)\cdot \mathbf{T}_1, 
    \bigtriangledown f_{\mathbf{b}}(\mathbf{p}_2)\cdot \mathbf{T}_2, \ldots 
,\bigtriangledown f_{\mathbf{b}}(\mathbf{p}_N)\cdot \mathbf{T}_N \right )^\top \right \| ^2.
\end{align*}
where $\bigtriangledown f_{\mathbf{b}}(\mathbf{p}_j)$ is the 
gradient vector of $f_{\mathbf{b}}(x,y)$ in any point $\mathbf{p}_j$. 
The role of $L_{WG}$ is to keep $\mathbf{p}_j$'s tangent vector and the gradient vector of 
the implicit polynomial $f_{\mathbf{b}}$ being perpendicular.
\begin{theorem}
The weak gradient constraint $L_{WG}$ can be written in 
a homogeneous quadratic form of $\mathbf{b}$.
\end{theorem}
\begin{proof}
Notice that each inner product 
$\bigtriangledown f_{\mathbf{b}}(\mathbf{p}_j)\cdot \mathbf{T}_j$
can be represent as a linear combination of $\mathbf{b}$, $j=1,2,\ldots,N$.
Then
\begin{equation*}
    \left (\bigtriangledown f_{\mathbf{b}}(\mathbf{p}_1)\cdot \mathbf{T}_1, 
    \bigtriangledown f_{\mathbf{b}}(\mathbf{p}_2)\cdot \mathbf{T}_2, \ldots 
,\bigtriangledown f_{\mathbf{b}}(\mathbf{p}_N)\cdot \mathbf{T}_N \right )^\top
=D_2\mathbf{b},
\end{equation*}
where $D_2$ 
is the collocation matrix whose rows are the coefficients of 
$\bigtriangledown f_{\mathbf{b}}(\mathbf{p}_j)\cdot \mathbf{T}_j$'s linear expressions.
Afterwards, $L_{WG}$ can be written in the matrix notation as
\begin{equation}\label{dis_lwg}
    L_{WG}=\mathbf{b}^\top A_2\mathbf{b},
\end{equation}
where 
\begin{equation*}
    A_2=D_2^\top D_2.
\end{equation*}
\end{proof}
% \begin{equation}
%     \left ( \begin{matrix}
% \bigtriangledown f_{\mathbf{b}}(\mathbf{p}_1)\cdot \mathbf{T}_1
%  \\\bigtriangledown f_{\mathbf{b}}(\mathbf{p}_2)\cdot \mathbf{T}_2
%  \\\vdots 
%  \\\bigtriangledown f_{\mathbf{b}}(\mathbf{p}_N)\cdot \mathbf{T}_N

% \end{matrix} \right )
% =\left (\begin{matrix}
%   \times & \times  & \cdots  & \times \\
%   \times & \times  & \cdots  & \times \\
%   \vdots & \vdots  &  & \vdots\\
%   \times & \times  & \cdots  & \times \\
% \end{matrix}  \right )
% \left ( \begin{matrix}b_1
%  \\b_2
%  \\\vdots
%  \\b_k
% \end{matrix} \right )  
% =D_2\mathbf{b},
% \end{equation}
Finally, due to the Equation (\ref{dis_lad}) and (\ref{dis_lwg}),
the discrete approximate
implicitization is found by minimizing 
the positive semidefinite quadratic objective function
\begin{equation}
    L_{\lambda,n}(\mathbf{b})=
    L_{AD} +\lambda L_{WG}
    =\mathbf{b}^\top (A_1+ \lambda A_2) \mathbf{b}
\end{equation}
over the coefficients $\mathbf{b}$ of
$f_{\mathbf{b} }(x,y)$,  while keeping the 
degree $n$ of $f_{\mathbf{b} }(x,y)$ fixed. 
Denote by $\mathbf{b}_{WGM}$ the unit eigenvector 
corresponding to the smallest eigenvalue of
\begin{equation*}
    A=A_1+ \lambda A_2,
\end{equation*}
then $\mathbf{b}_{WGM}$ is the solution for minimizing 
$L_{\lambda,n}(\mathbf{b})$
subject to $\left \| \mathbf{b} \right \| =1$.
\subsection{Adaptive implicitization algorithm}
The weak gradient method (WGM for short) of adaptive implicitization
for non-polynomial curves is summarized in Algorithm 2. 
Algorithm 2 is almost identical to
Algorithm 1, with three changes:
\begin{itemize}
    \item the curve sampling is employed in the first place (line 1);
    \item the AD and WG matrices $A_1^{(n)}, A_2^{(n)}$ are 
    constructed without the Gram matrix $G_{\alpha}$ (line 6);
    \item since the change trend of the WG’s loss for non-polynomial 
    curves is more subtle (see Figure \ref{error4} for example),
    the stopping criterion of $\epsilon_{WG}$ is replaced
    by $e_2^{(n)} \le \epsilon_{WG}$ (line 14),
    to check whether the WG's loss satisfies our default precision.
\end{itemize}
\begin{algorithm}
\caption{WGM for non-polynomial curves}\label{alg:two}
\KwData{The parametric curve $\mathbf{p}(t)$;
the maximum implicit degree $n_{\max}$;
the stopping thresholds $\epsilon_{AD},\epsilon_{WG}$.}
\KwResult{The coefficient vector $\mathbf{b}_{WGM}$ of the implicit 
polynomial $f_{\mathbf{b}}$.}
Uniform sampling on $\mathbf{p}(t)$\;
$n\leftarrow 1$\;
\While{$n\leq n_{\max}$}{
Construct the collocation matrix $D_1^{(n)}$ 
of the AD constraint\;
Construct the collocation matrix $D_2^{(n)}$ 
of the WG constraint\;
$A_1^{(n)}, A_2^{(n)}\leftarrow 
(D_1^{(n)})^\top D_1^{(n)}, (D_2^{(n)})^\top D_2^{(n)}$\;
% $\eta_{\min}^{(m)}\leftarrow$ the smallest
% eigenvalue of $A_2^{(m)}$\;
$A^{(n)}\leftarrow A_1^{(n)}+\lambda A_2^{(n)}$\;
$\mathbf{b}^{(n)}\leftarrow$ the unit eigenvector 
corresponding to the smallest eigenvalue of $A^{(n)}$\;
$e_1^{(n)}, e_2^{(n)}\leftarrow
(\mathbf{b}^{(n)})^\top A_1^{(n)}\mathbf{b}^{(n)},
(\mathbf{b}^{(n)})^\top A_2^{(n)}\mathbf{b}^{(n)}$\;
  \eIf{$n=n_{\max}$}{
  $\mathbf{b}_{WGM}\leftarrow\mathbf{b}^{(n)}$\;
  \Return{} %\Comment*[r]{This is a comment}
  }{\If{$e_1^{(n)}\le\epsilon_{AD}$
and $e_2^{(n)}
\le \epsilon_{WG}$}{
  $\mathbf{b}_{WGM}\leftarrow\mathbf{b}^{(n)}$\;
  \Return{}
    }
  }
$n\leftarrow n+1$\;
}
\end{algorithm}

\subsection{Experiments}

We choose the basis $\left \{ \phi_i(x,y) \right \}_{i=1}^{k}$ to be 
the bivariate monomial basis of total degree $n$.
We set the maximum implicit degree $n_{\max}=7$, the regulator gain $\lambda=0.01$, and the thresholds $\epsilon_{AD}=10^{-2}$ and $\epsilon_{WG}=10^{-1}$.

\begin{example}
Consider the non-polynomial parametric curve
\begin{equation*}
C_3(t)=\left ( \begin{matrix}
 2(1+\cos t)\cos t\\
2(1+\cos t)\sin t
\end{matrix} \right ),
\end{equation*}
where the parameter of $C_3(t)$ take values in 
$[0, 10]$.
$C_3(t)$ is shown in Figure \ref{input_ex3}.
\begin{figure}[h]%
\centering
\includegraphics[width=0.5\textwidth]{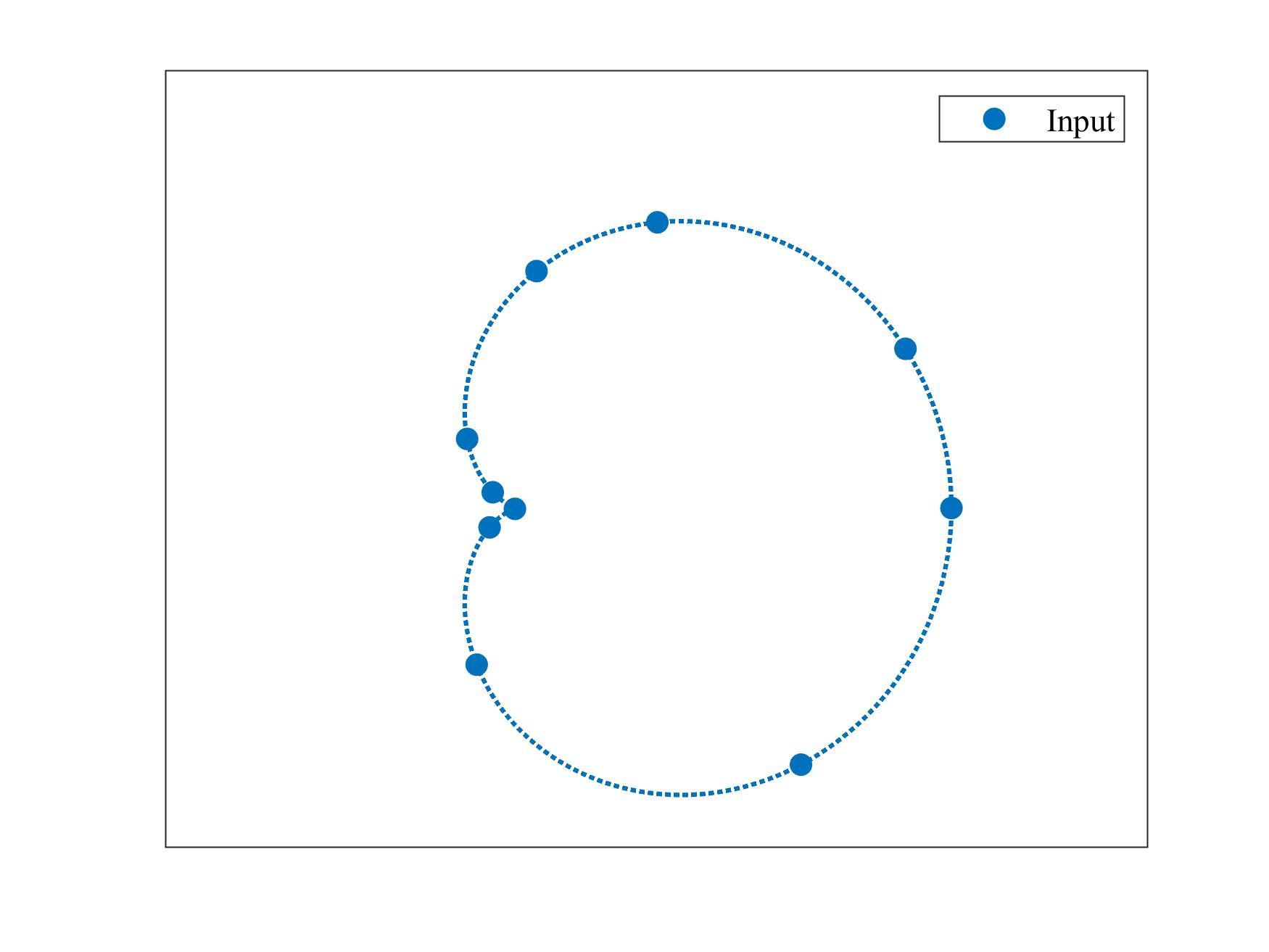}
\caption{The input curve $C_3(t)$ and $10$ points are sampled uniformly from $C_3(t)$.}\label{input_ex3}
\end{figure}
\begin{sidewaysfigure}[htbp]
	\centering
	\subfigure[WGM ($n=3$)]{\begin{minipage}[c]{0.25\textwidth}
		\centering
		\includegraphics[width=\textwidth]{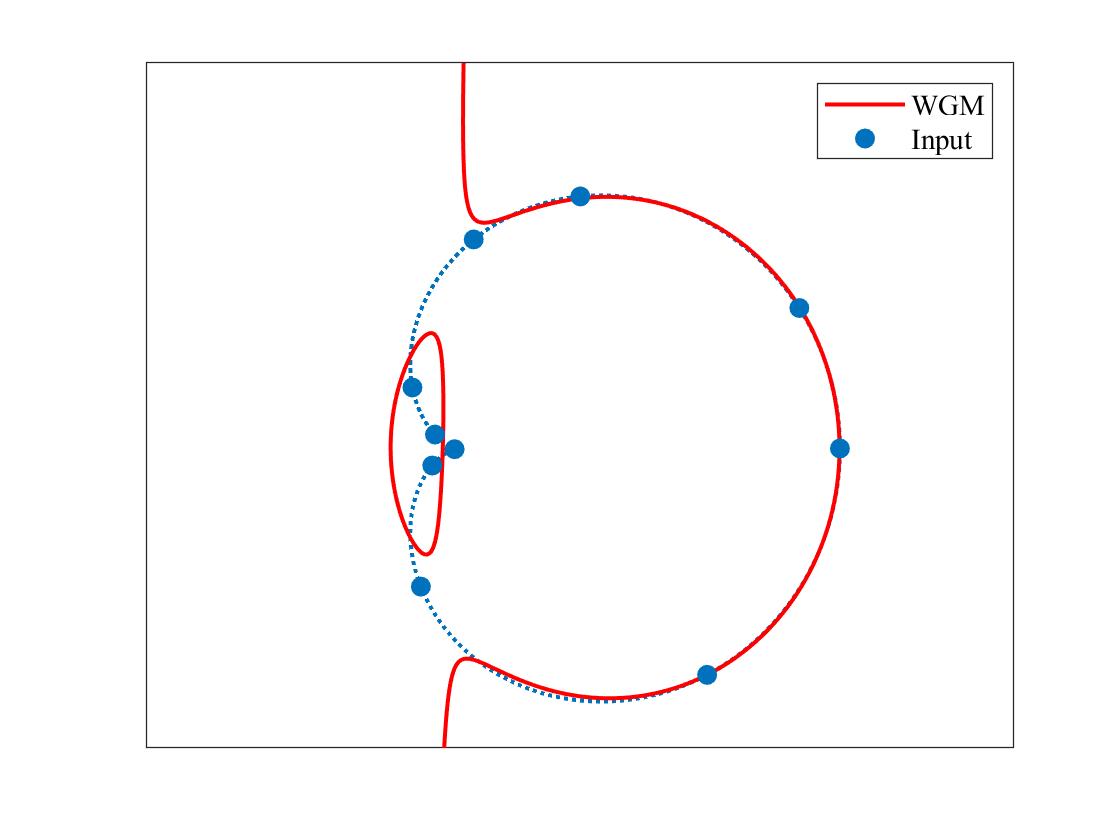}
	\end{minipage}}
	 \subfigure[WGM ($n=4$)]{\begin{minipage}[c]{0.25\textwidth}
		\centering
		\includegraphics[width=\textwidth]{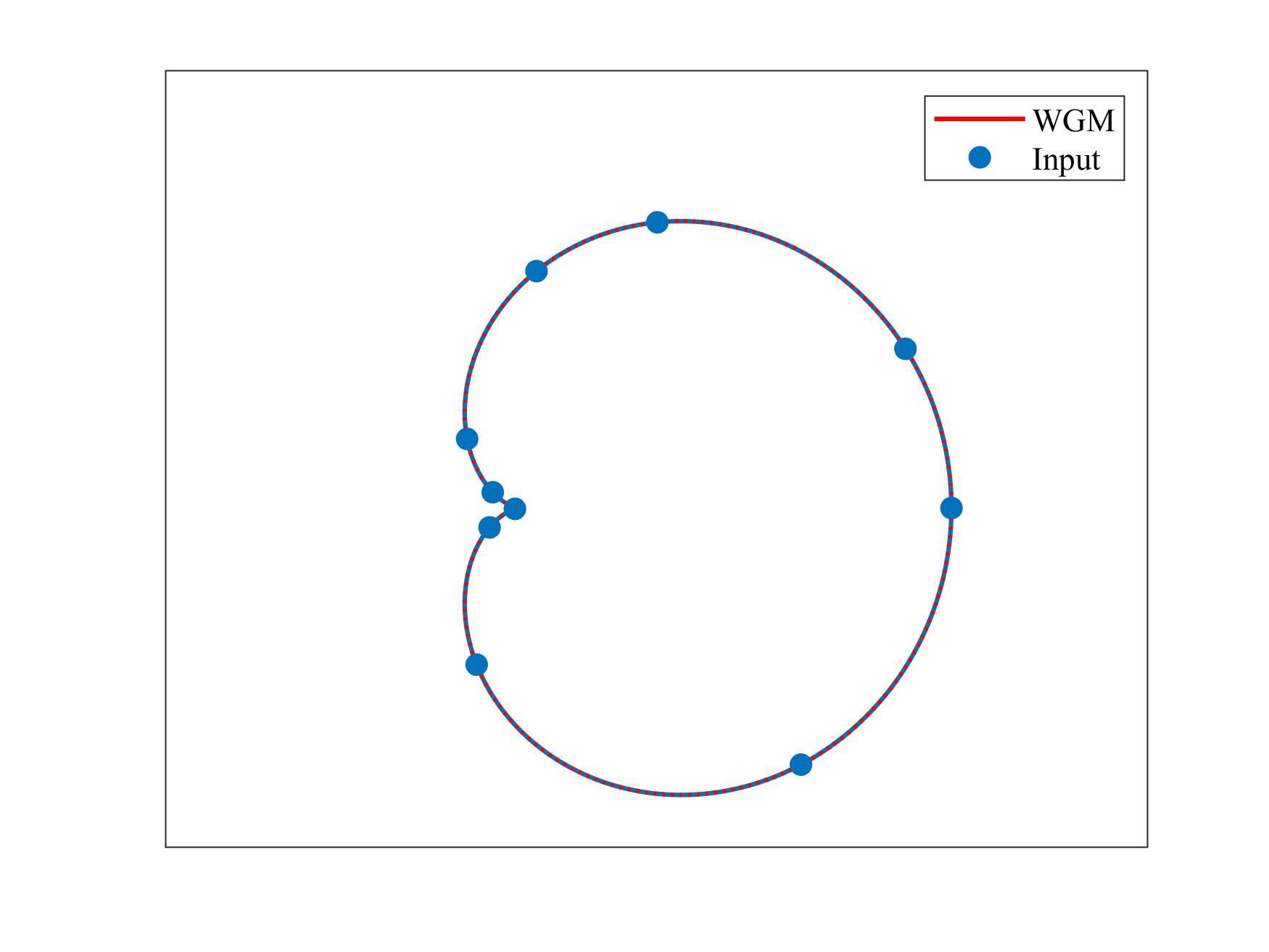}
	\end{minipage} }
	\subfigure[WGM ($n=5$)]{\begin{minipage}[c]{0.25\textwidth}
		\centering
		\includegraphics[width=\textwidth]{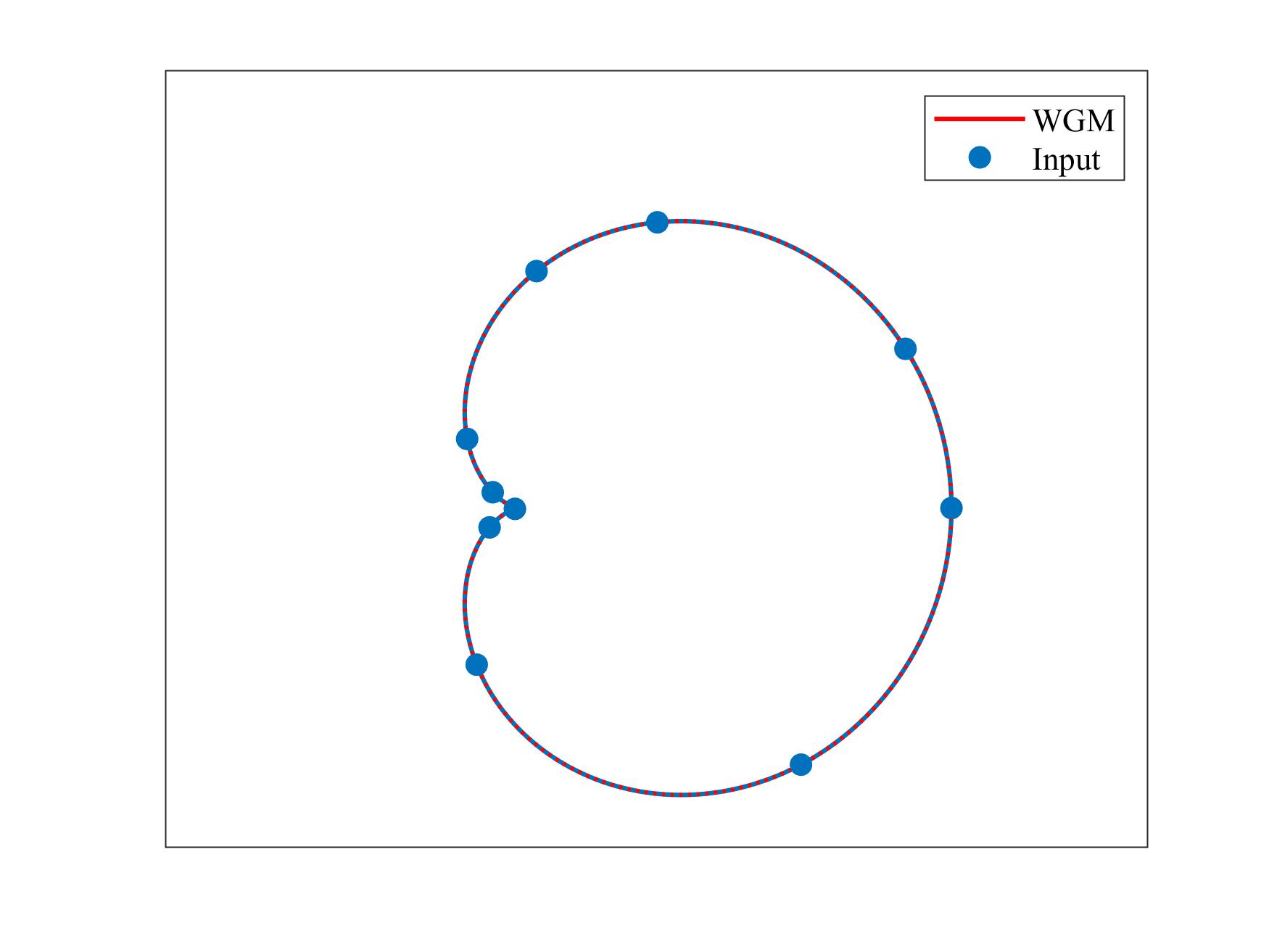}
	\end{minipage}}
	\\
	\subfigure[DM ($n=3$)]{\begin{minipage}[c]{0.25\textwidth}
		\centering
		\includegraphics[width=\textwidth]{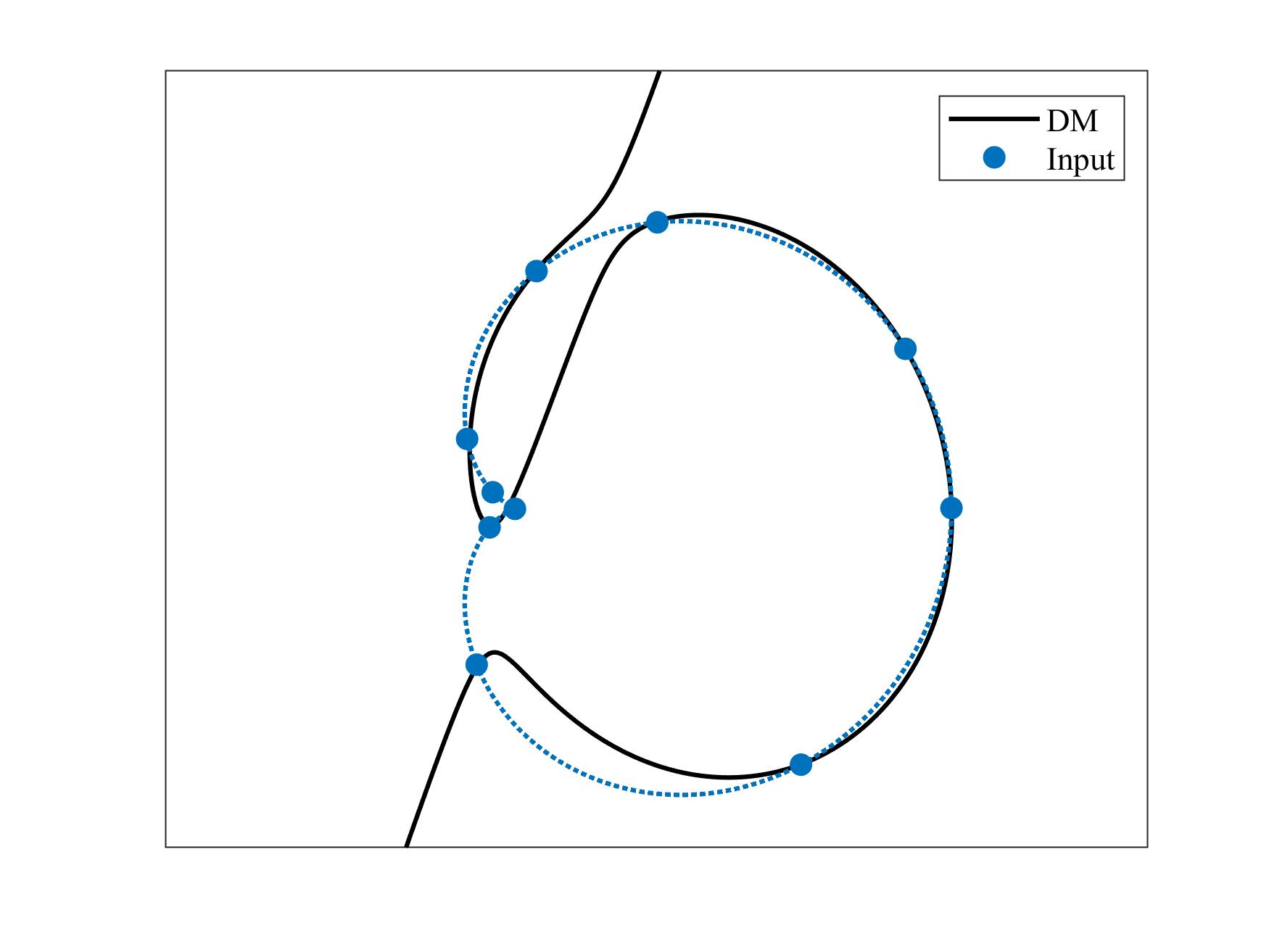}
	\end{minipage}}
	 \subfigure[DM ($n=4$)]{\begin{minipage}[c]{0.25\textwidth}
		\centering
		\includegraphics[width=\textwidth]{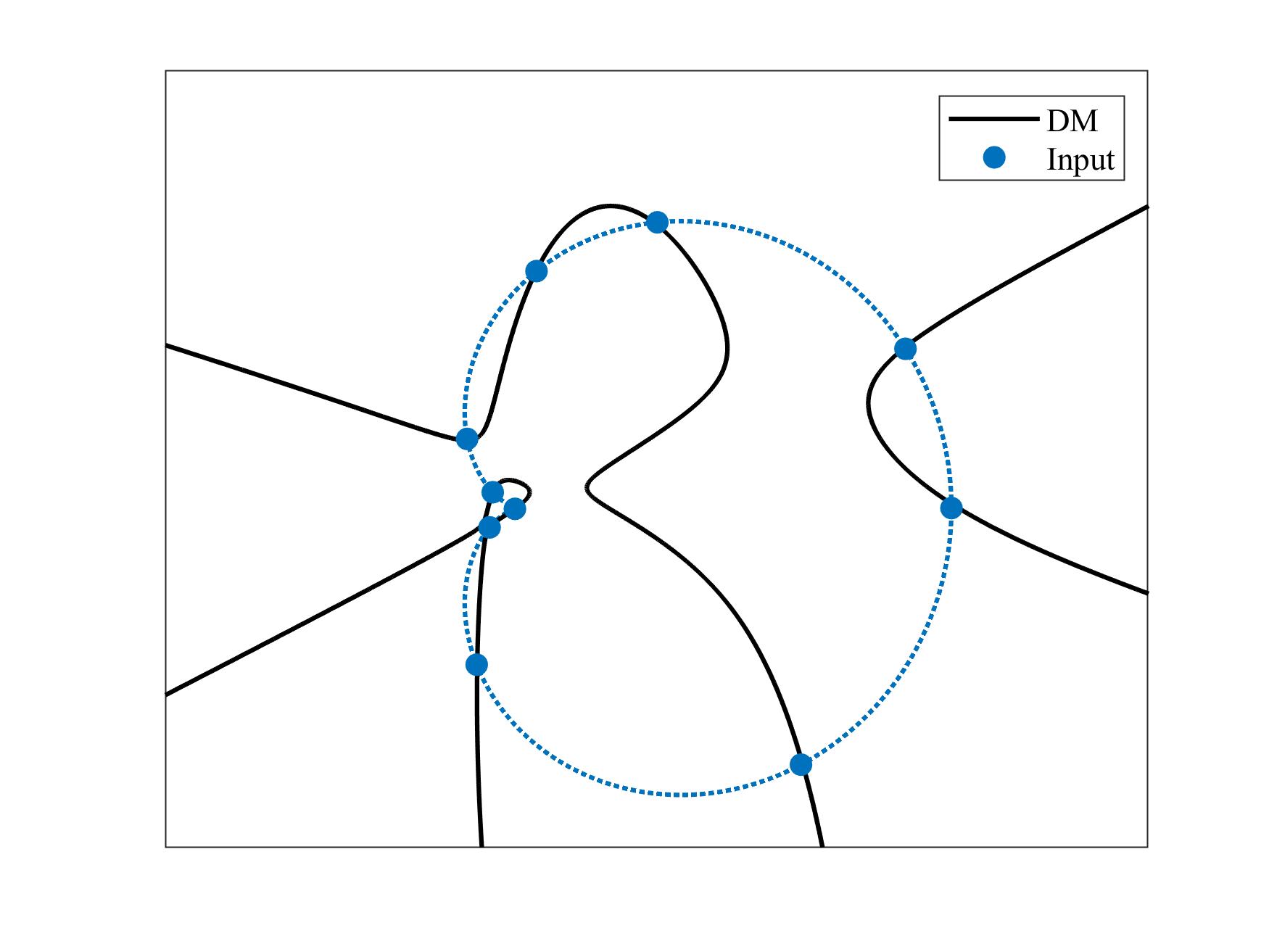}
	\end{minipage} }
	\subfigure[DM ($n=5$)]{\begin{minipage}[c]{0.25\textwidth}
		\centering
		\includegraphics[width=\textwidth]{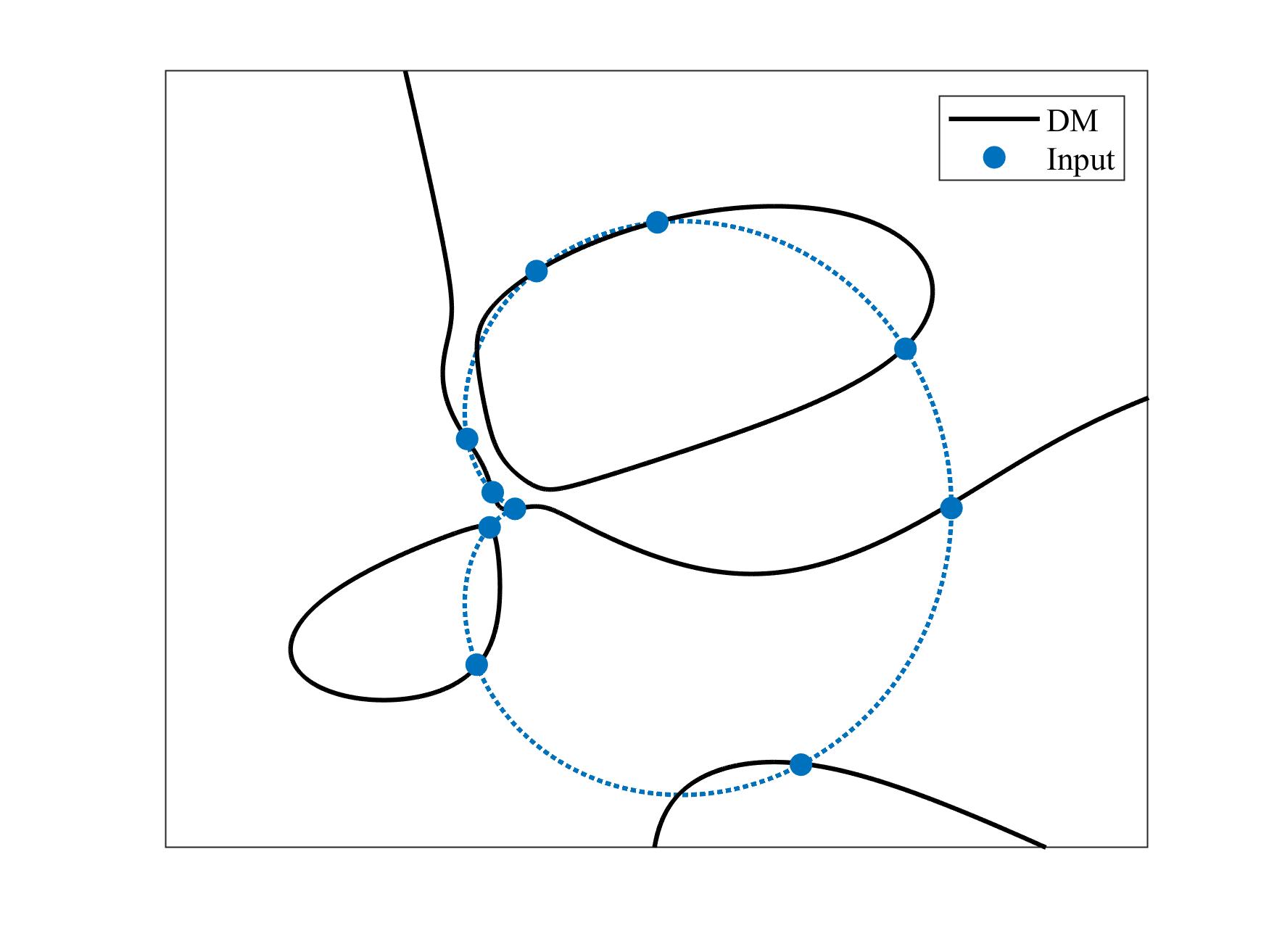}
	\end{minipage}}
	\caption{Adaptive implicitization of $C_3(t)$.
	The blue dash line in (a)-(f) is the input curve, 
	the red line in (a)-(c) is the output curve by our method, 
	and the black line in (d)-(f) is the output curve by Dokken's method. 
	From left to right: the implicit degree $n=3,4,5$.}
	\label{ex3_output}
\end{sidewaysfigure}

%----------------------------------------------------------
\begin{figure}
    \centering
    \includegraphics[width=0.5\textwidth]{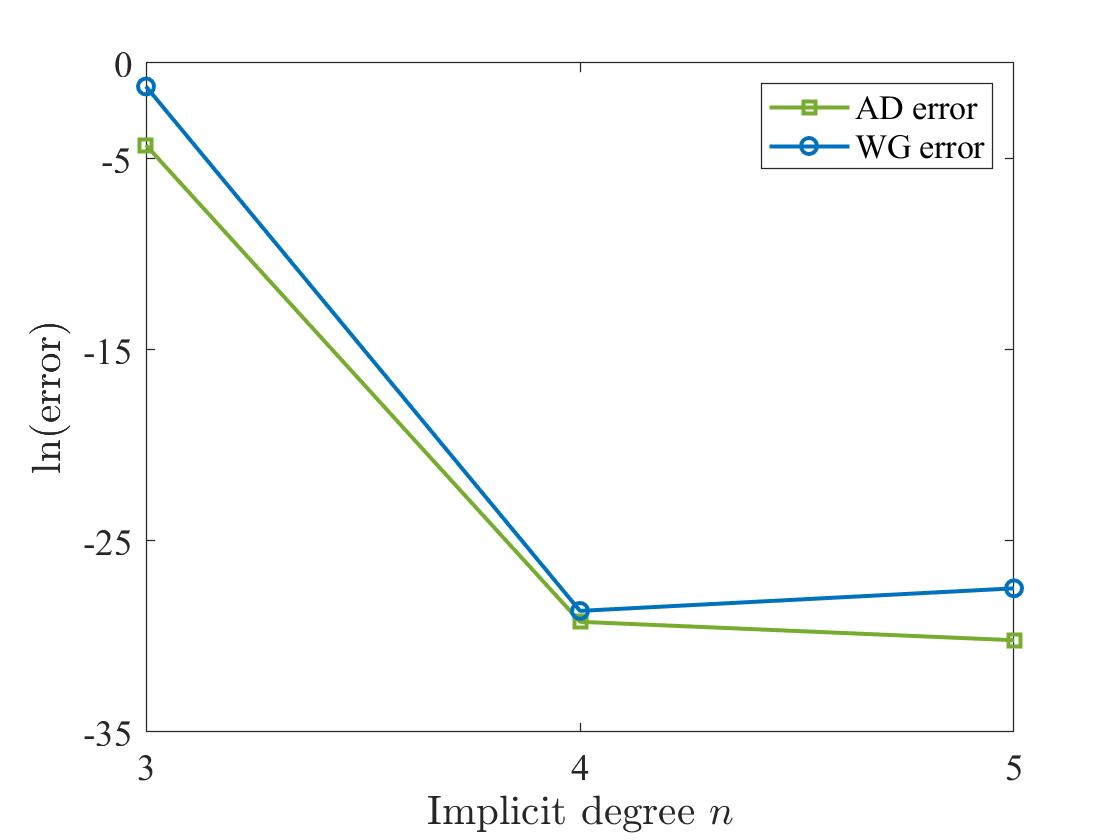}
    \caption{Statistics of our method on changes of the AD and WG error for $C_3(t)$,
    as the implicit degree $n$ increases.}
    \label{error3}
\end{figure} 

The first row of Figure \ref{ex3_output} shows the adaptive implicitization process 
of $C_3(t)$ by the WGM. Similarly, the second row of 
Figure \ref{ex3_output} shows the implicitization process
of $C_3(t)$ by Dokken's method.
We can see that for every iteration (i.e. implicit degree $n$), the WGM's output 
curve will approach $C_3(t)$ closer than that of DM, from the viewpoint of "shape-preserving".

Figure \ref{error3} shows the statistic of our method on changes 
of the AD and WG error for $C_3(t)$, when the implicit degree $n$ is increasing.
\end{example}
\begin{figure}[h]%
\centering
\includegraphics[width=0.5\textwidth]{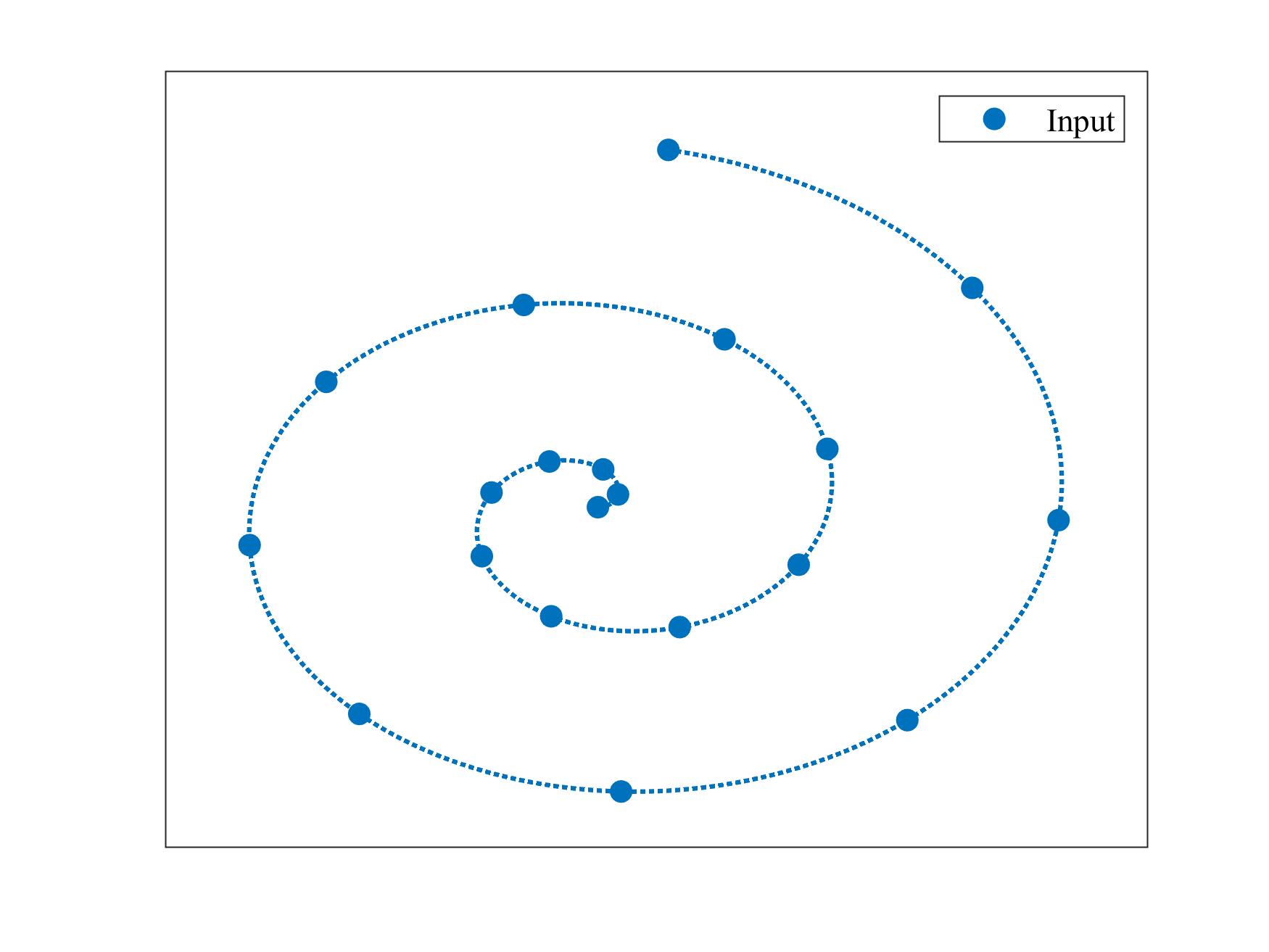}
\caption{The input curve $C_4(t)$ and $20$ points are sampled uniformly from $C_4(t)$.}\label{input_ex4}
\end{figure}
\begin{example}
Consider the non-polynomial parametric curve
\begin{equation*}
C_4(t)=\left ( \begin{matrix}
 t\cos t\\
t\sin t
\end{matrix} \right ),
\end{equation*}
where the parameters of $C_4(t)$ take values in $[0, 14]$.
$C_4(t)$ is shown in Figure \ref{input_ex4}.

The first row of Figure \ref{ex4_output} shows the adaptive implicitization process 
of $C_4(t)$ by the WGM. Similarly, the second row of 
Figure \ref{ex4_output} shows the implicitization process
of $C_4(t)$ by Dokken's method.
We can see that for every iteration, the WGM refrains from additional branches as much as possible. 

Figure \ref{error4} shows the statistic of our method on changes 
of the AD and WG error for $C_4(t)$, when the implicit degree $n$ is increasing.
\end{example}
%------------------------------------------------
\begin{sidewaysfigure}
	\subfigure[WGM ($n=3$)]{\begin{minipage}[c]{0.19\textwidth}
		\centering
		\includegraphics[width=\textwidth]{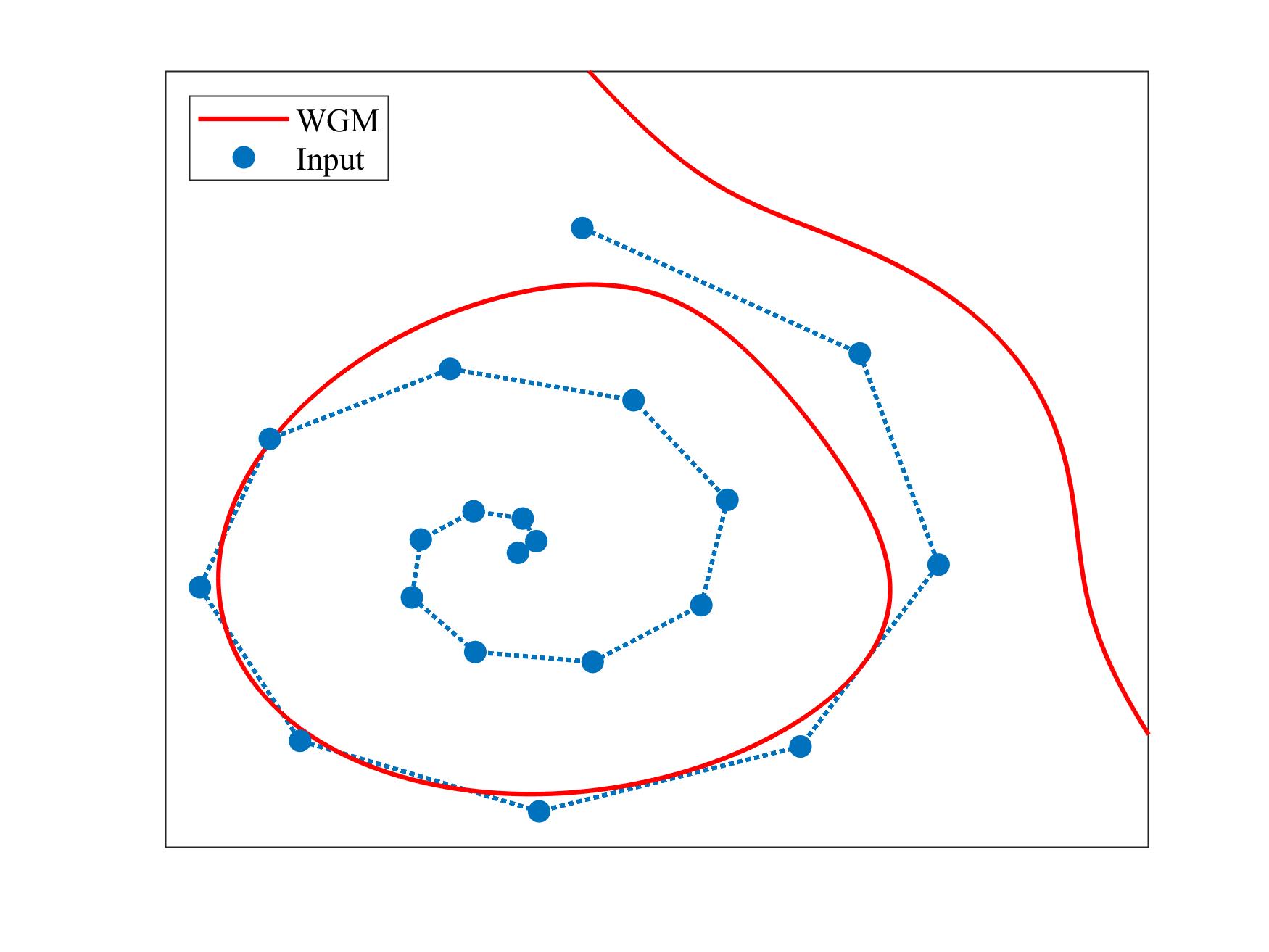}
	\end{minipage}}
	\subfigure[WGM ($n=4$)]{\begin{minipage}[c]{0.19\textwidth}
		\centering
		\includegraphics[width=\textwidth]{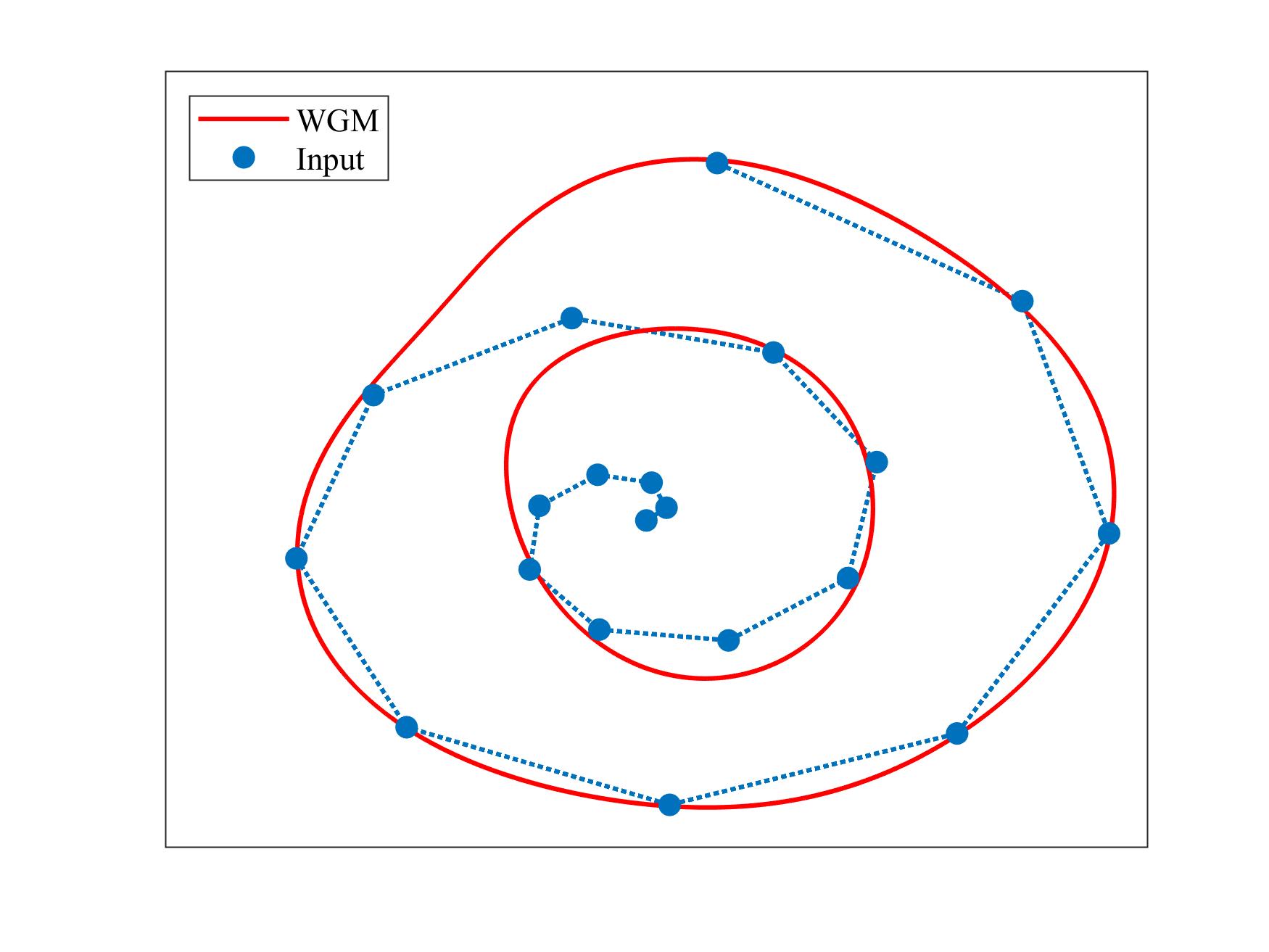}
	\end{minipage}}
	\subfigure[WGM ($n=5$)]{\begin{minipage}[c]{0.19\textwidth}
		\centering
		\includegraphics[width=\textwidth]{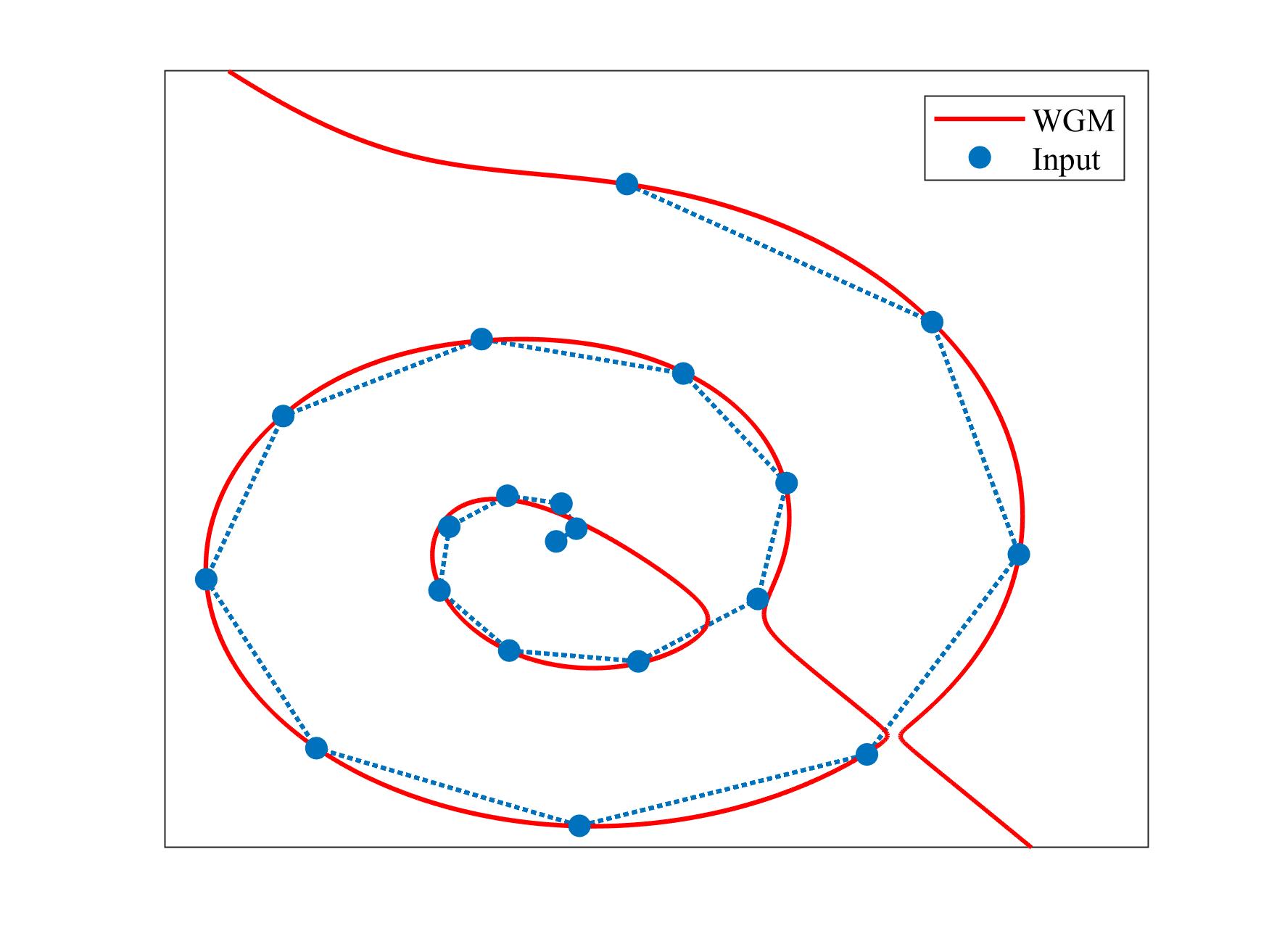}
	\end{minipage}}
	\subfigure[WGM ($n=6$)]{\begin{minipage}[c]{0.19\textwidth}
		\centering
		\includegraphics[width=\textwidth]{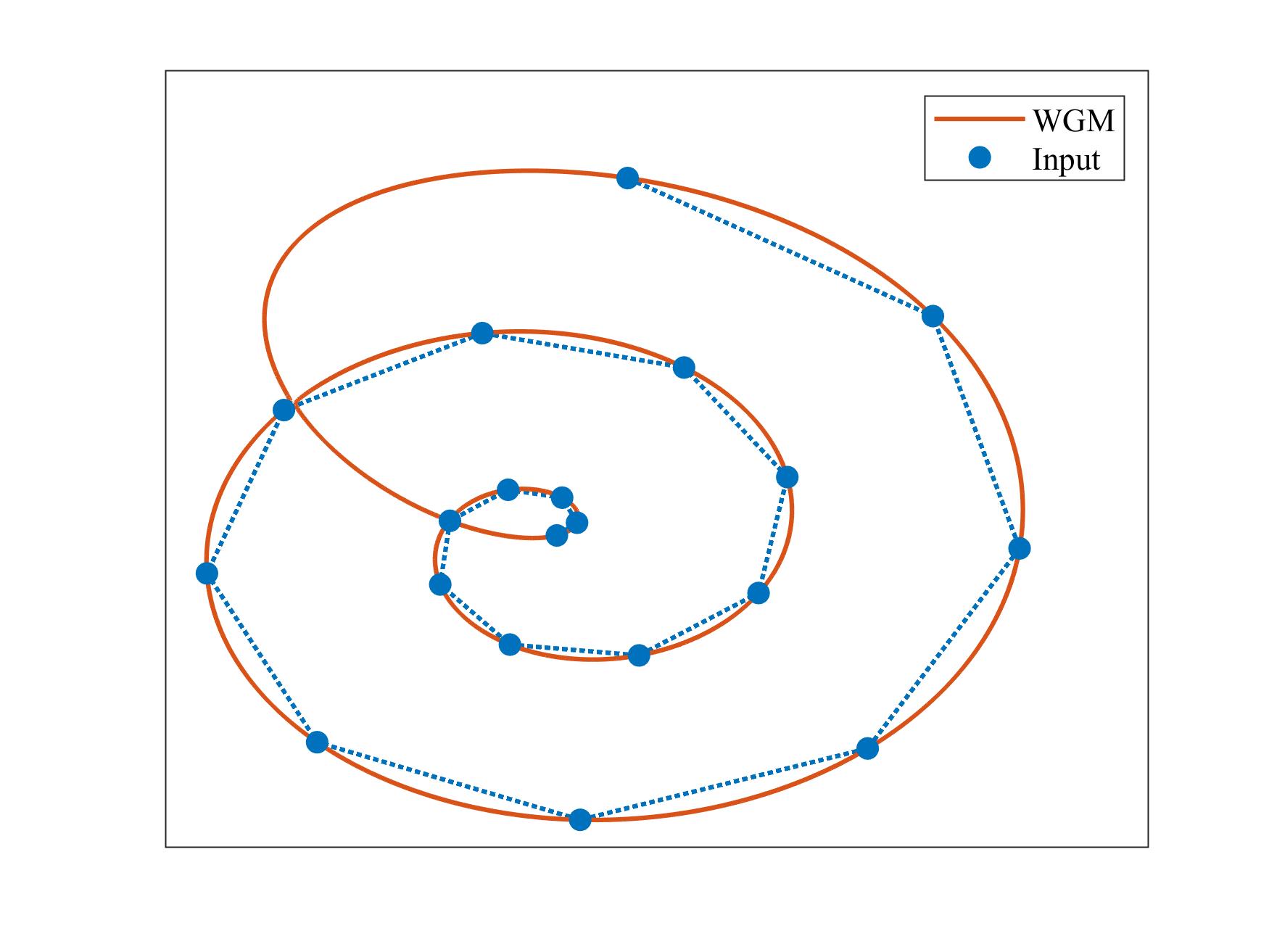}
	\end{minipage}}
	\subfigure[WGM ($n=7$)]{\begin{minipage}[c]{0.19\textwidth}
		\centering
		\includegraphics[width=\textwidth]{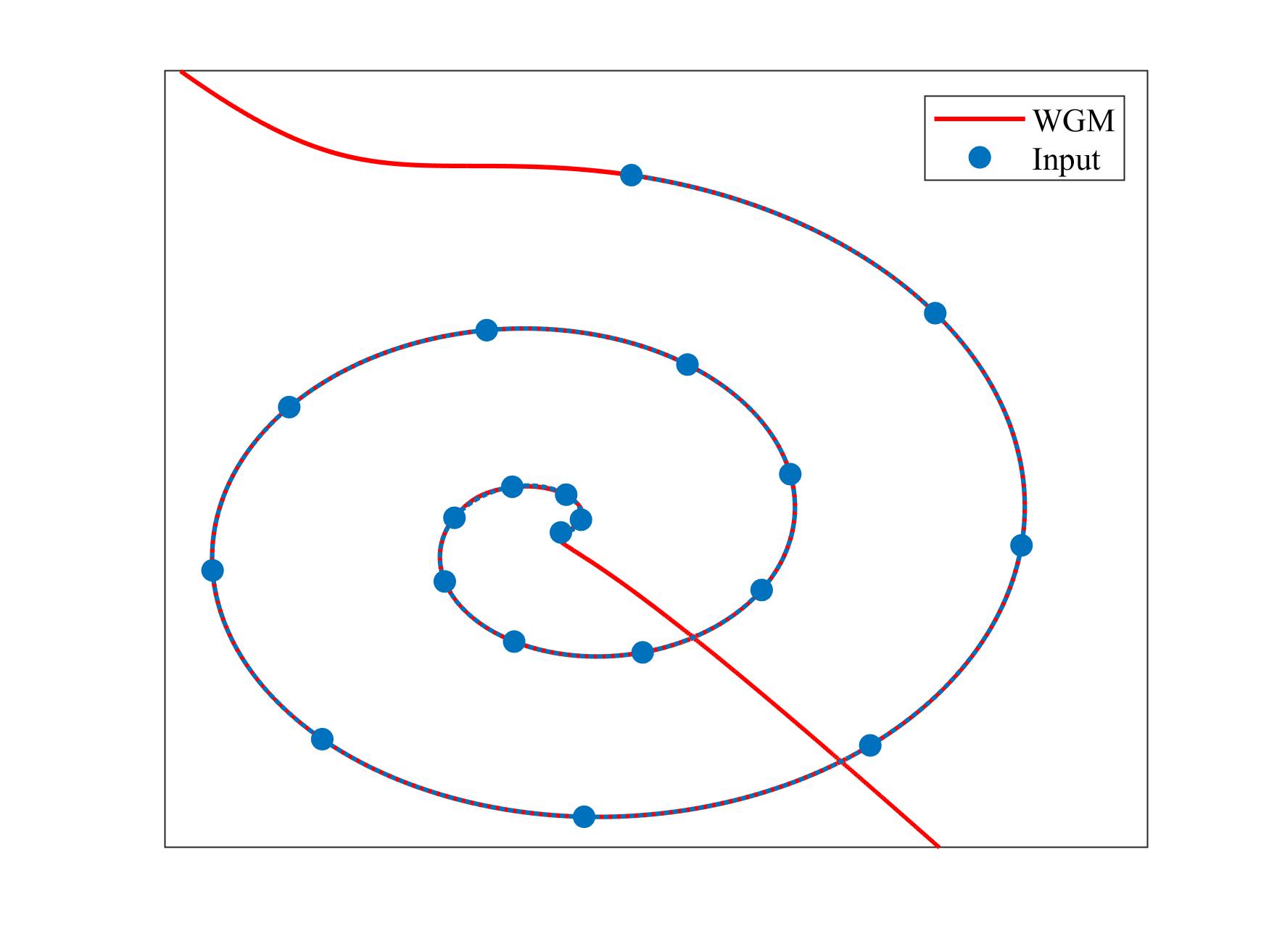}
	\end{minipage}}
	\\
	\subfigure[DM ($n=3$)]{\begin{minipage}[c]{0.19\textwidth}
		\centering
		\includegraphics[width=\textwidth]{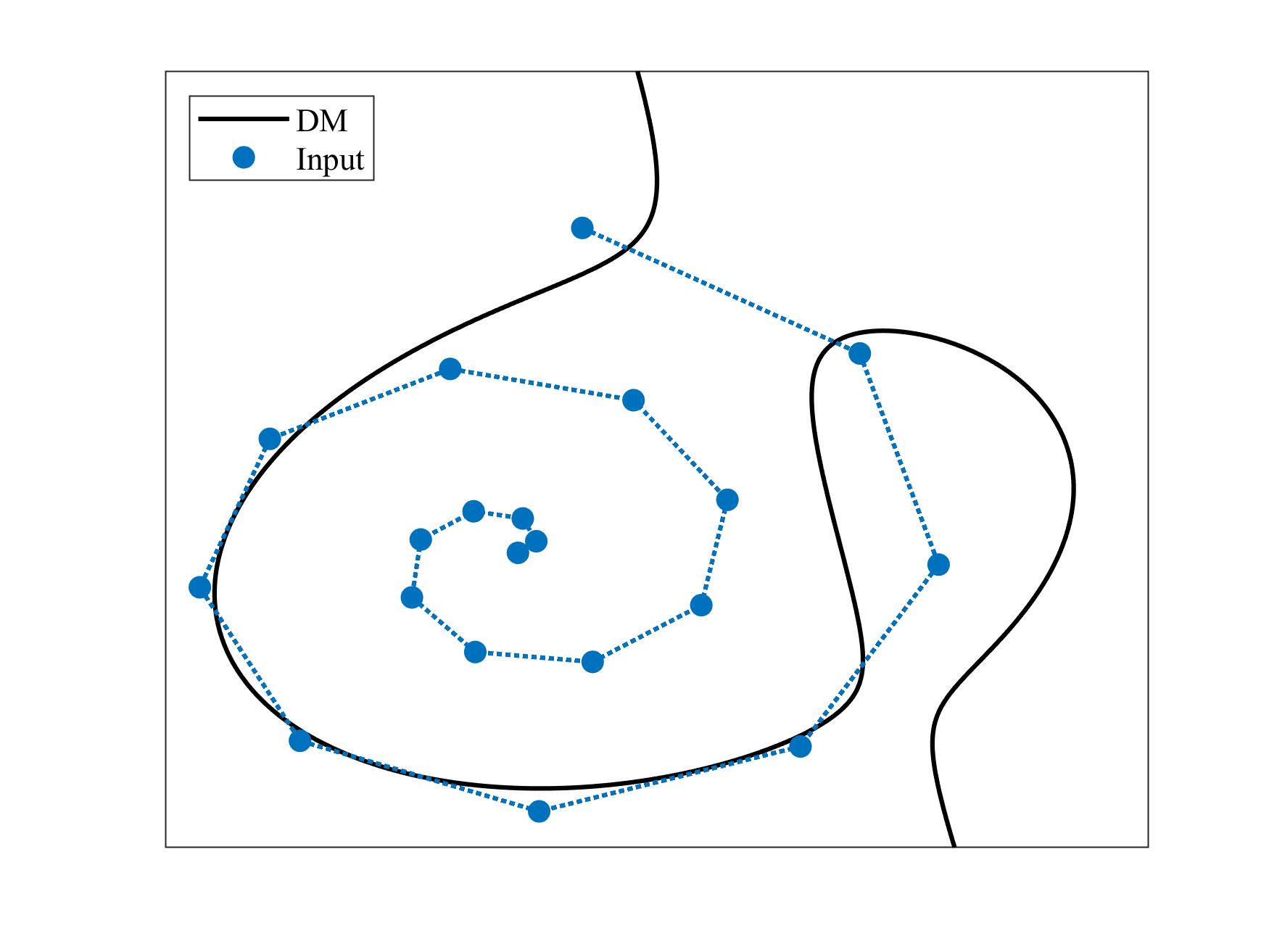}
	\end{minipage}}
	\subfigure[DM ($n=4$)]{\begin{minipage}[c]{0.19\textwidth}
		\centering
		\includegraphics[width=\textwidth]{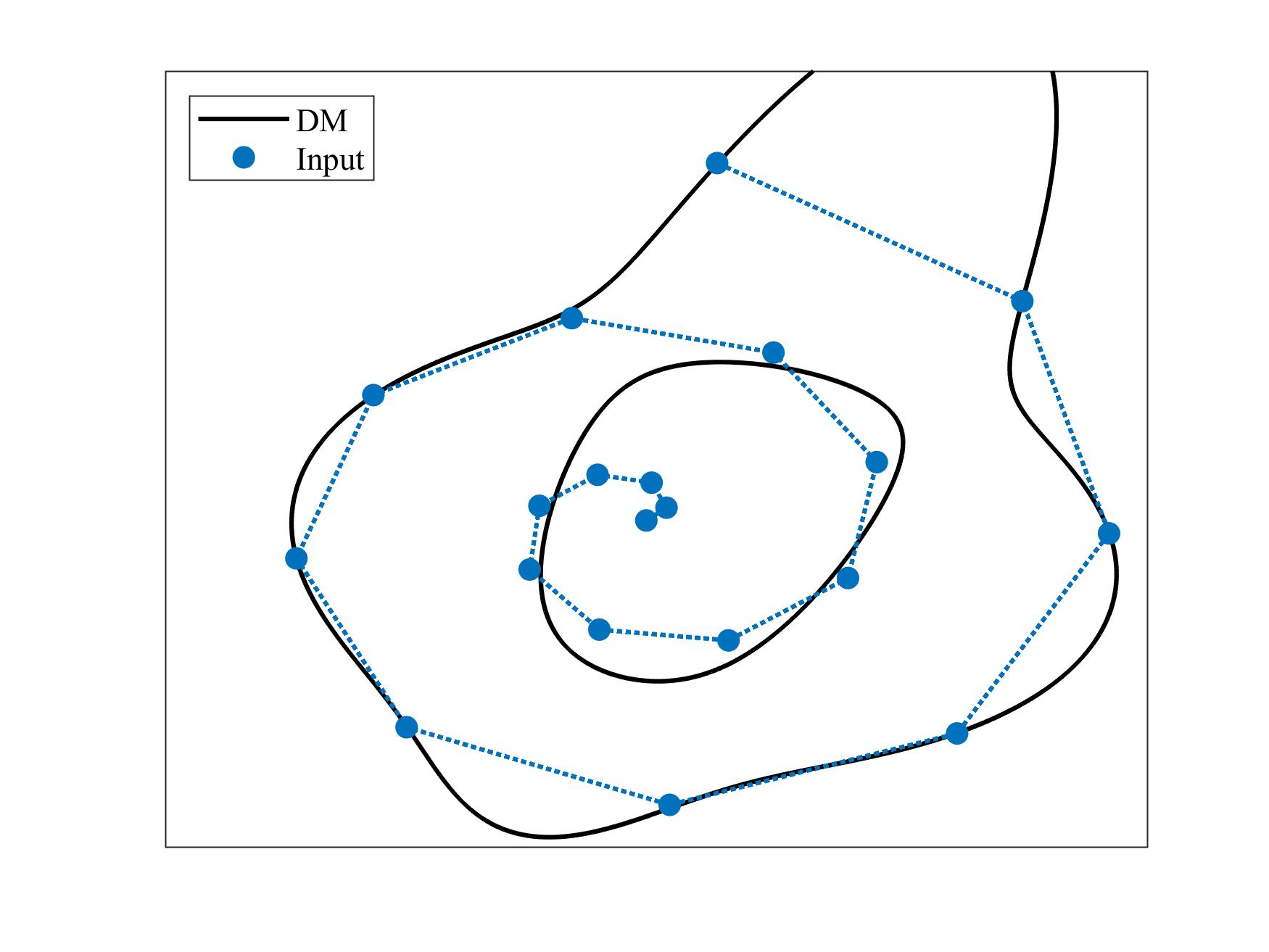}
	\end{minipage}}
	\subfigure[DM ($n=5$)]{\begin{minipage}[c]{0.19\textwidth}
		\centering
		\includegraphics[width=\textwidth]{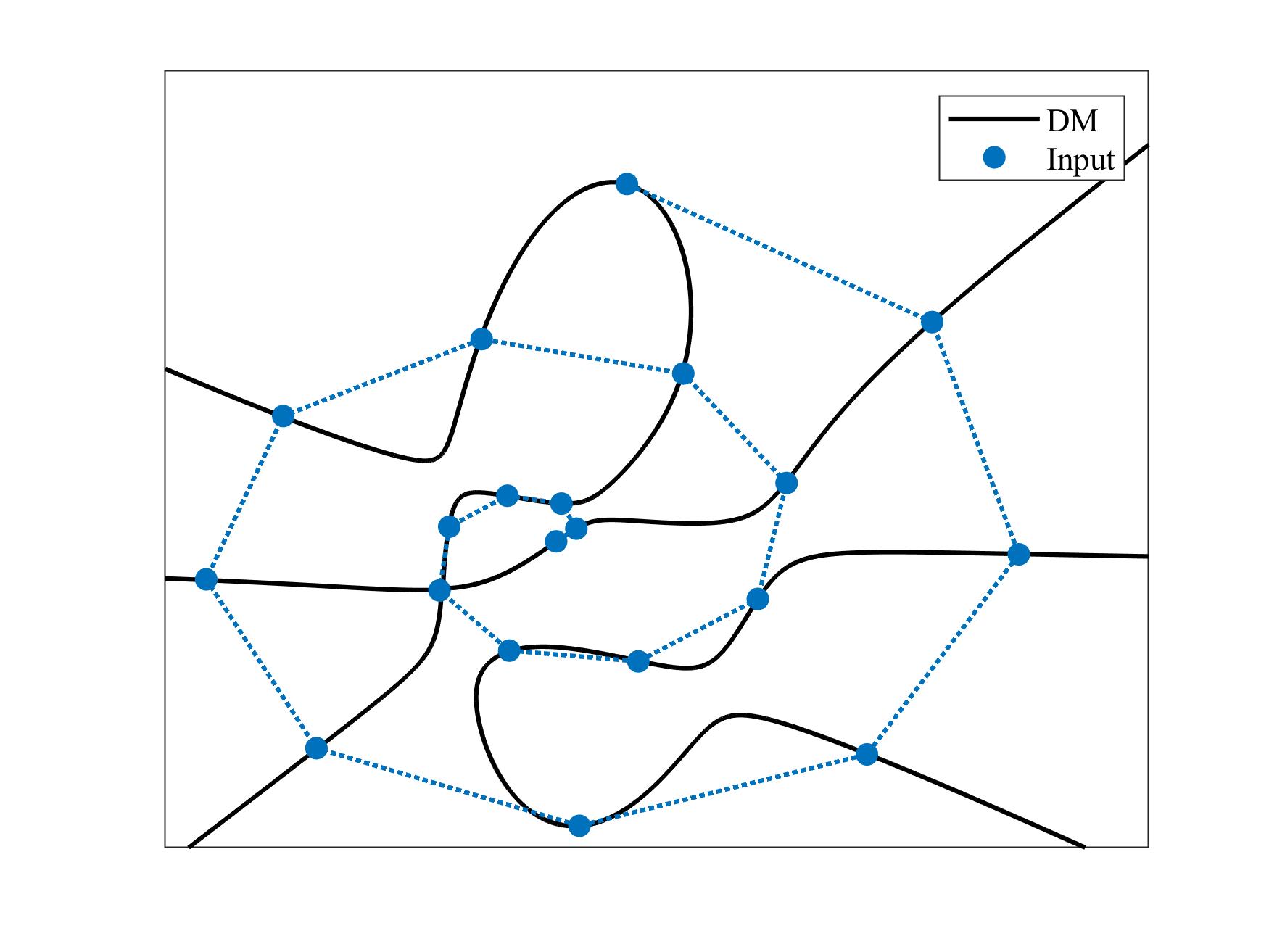}
	\end{minipage}}
	\subfigure[DM ($n=6$)]{\begin{minipage}[c]{0.19\textwidth}
		\centering
		\includegraphics[width=\textwidth]{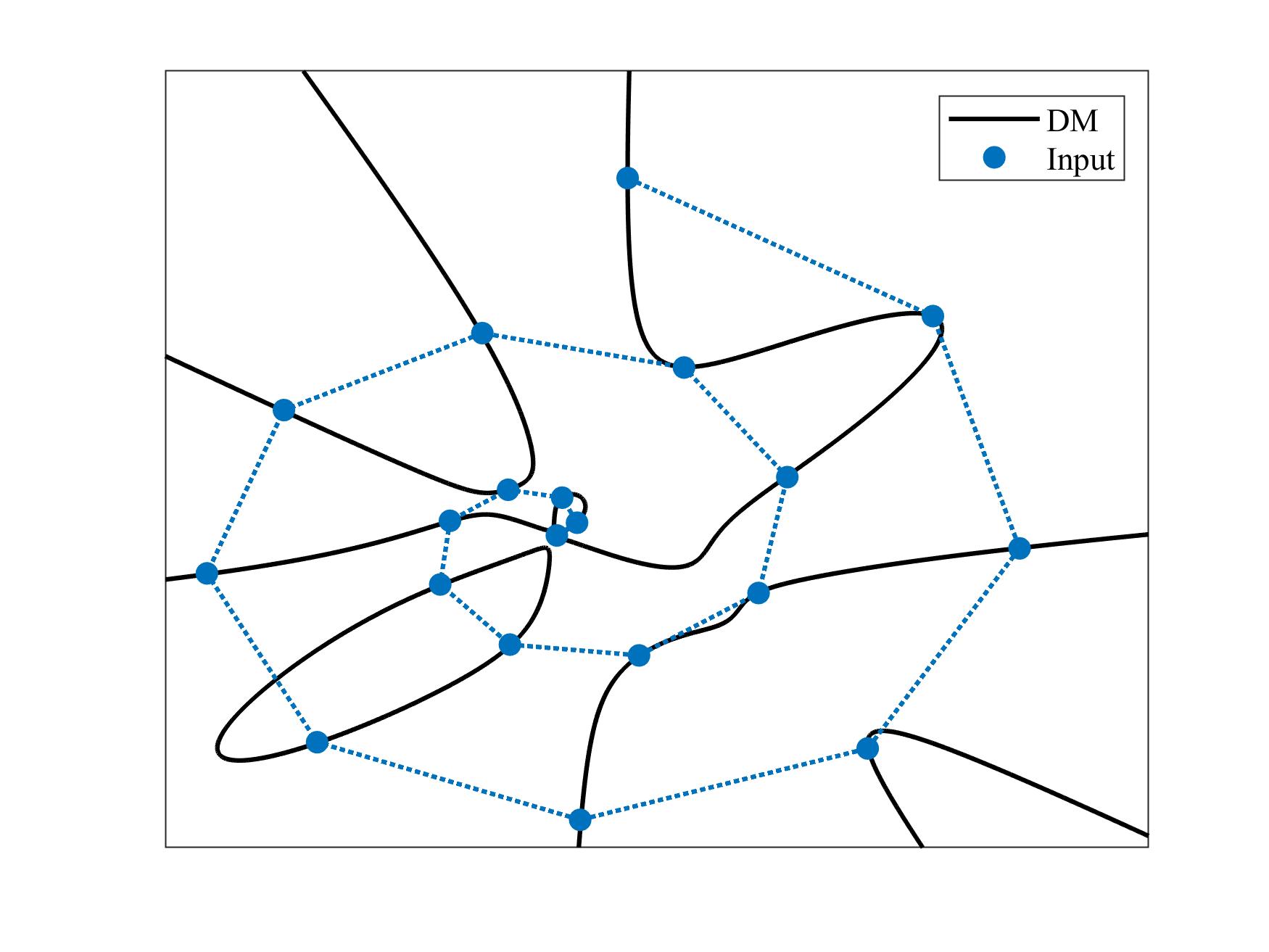}
	\end{minipage}}
	\subfigure[DM ($n=7$)]{\begin{minipage}[c]{0.19\textwidth}
		\centering
		\includegraphics[width=\textwidth]{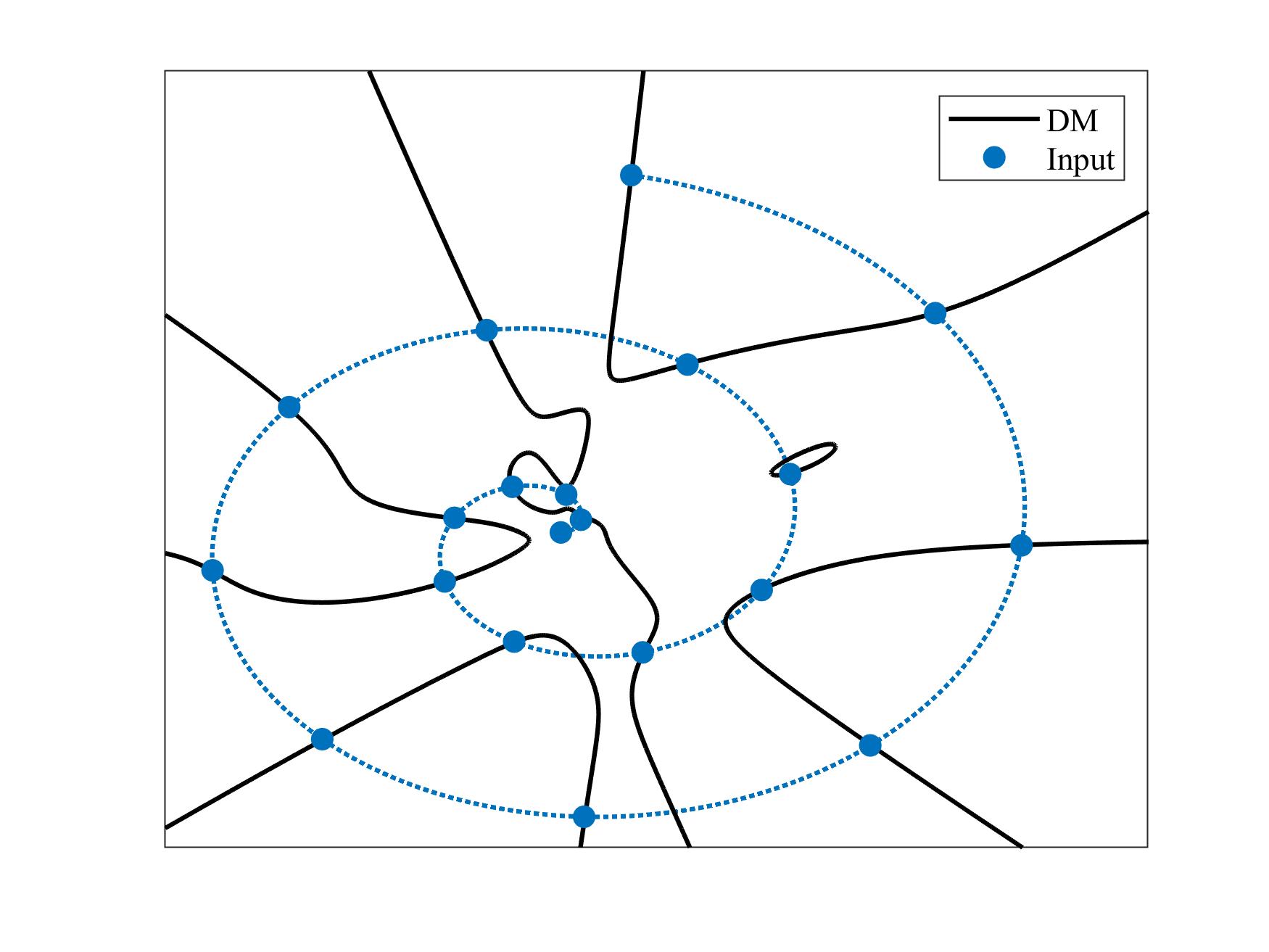}
	\end{minipage}}
	\caption{Adaptive implicitization of $C_4(t)$.
	The blue dash line in (a)-(j) is the input curve, 
	the red line in (a)-(c), (g), and (h) is the output curve by our method, 
	and the black line in (d)-(f), (i) and (j) is the output curve by Dokken's method. 
	From left to right: the implicit degree $n=3,4,5$.}
	\label{ex4_output}
\end{sidewaysfigure}
%-------------------------------------------------------
\begin{figure}
    \centering
    \includegraphics[width=0.5\textwidth]{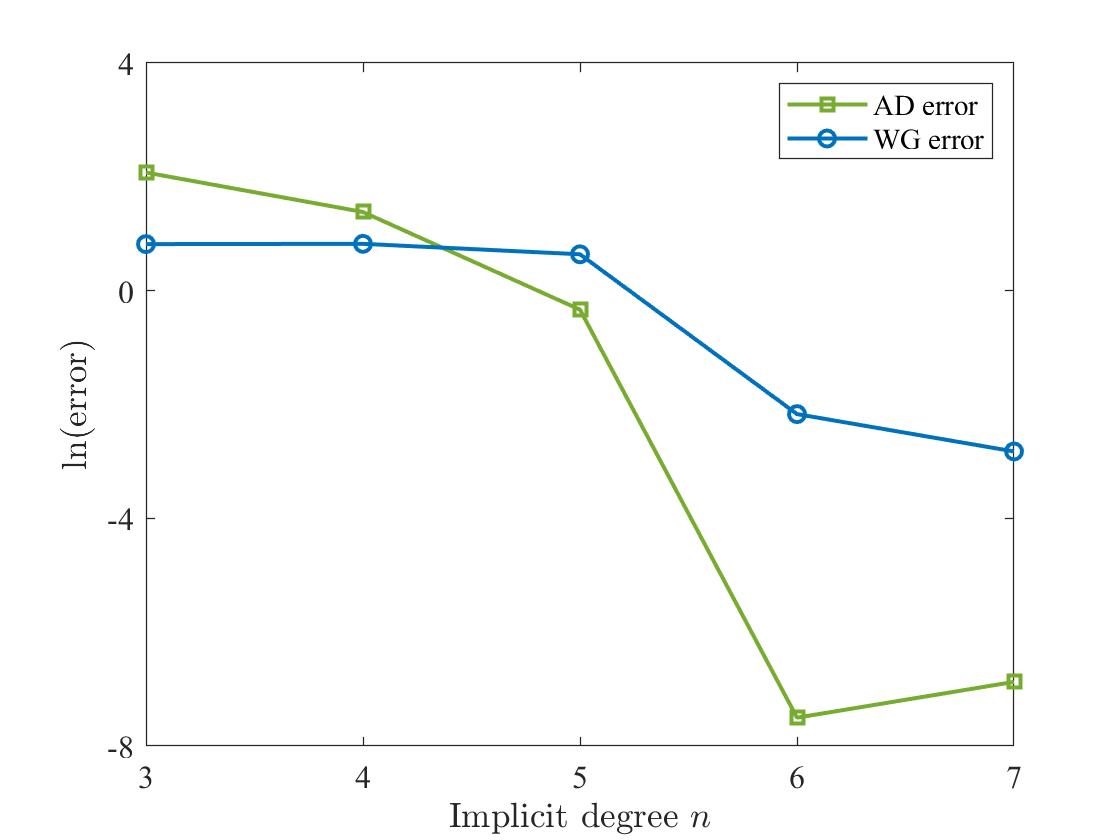}
    \caption{Statistics of our method on changes of the AD and WG error for $C_4(t)$,
    as the implicit degree $n$ increases.}
    \label{error4}
\end{figure} 

Finally, Table \ref{tab2} shows the comparison of running time performance of WGM, DM 
and the method in \cite{interian2017curve}.
\begin{table}[h!]
  \begin{center}
    \caption{Performance of WGM on $C_3(t), C_4(t)$ (all timings are measured in seconds).}
    \label{tab2}
    \begin{tabular}{cccccc} 
      \hline
      \hline
Input	& AD Error	& WG Error 
& WGM & DM & Method in \cite{interian2017curve}\\
 	&  	&   
& Time  & Time  & Time \\
      \hline
$C_3(t)$		& $2.011e-13$			& $3.558e-13$  & $0.0039$  & $0.0027$ & $59.49$\\
$C_4(t)$		& $1.114e-3$			& $6.721e-2$  & $0.0051$  & $0.0039$ & $166.58$\\
          \hline
    \end{tabular}
  \end{center}
\end{table}
\section{Conclusion}\label{conclusion}
In this paper, we proposed a novel approach for 
adaptive implicitization of parametric curves based on the so-called weak geometric constraint, 
named WGM. WGM solves the implicitization problem with regularization 
terms naturally with very little extra computation effort. Thus, it not only avoids 
additional branches but also reduces the computational cost effectively. 
Several experiments presented demonstrate that WGM produces 
high-quality implicitization results. 
In future work, we plan to generalize the proposed method to the cases of parametric surfaces
and space parametirc curves.

%Bibliography
\bibliographystyle{unsrt}  
\bibliography{references}

\end{document}